\documentclass[11pt]{article}
\usepackage[left=1in,right=1in,top=1in,bottom=1in]{geometry}
\usepackage{graphicx}
\usepackage[normalem]{ulem} 
\usepackage[utf8]{inputenc} 
\usepackage[T1]{fontenc}    
\usepackage{hyperref}       
\usepackage{url}            
\usepackage{booktabs}       
\usepackage{mathtools}
\usepackage{amsfonts}       
\usepackage{nicefrac}       
\usepackage{microtype}      
\usepackage{comment}
\usepackage{graphicx, ulem}
\usepackage[dvipsnames]{xcolor}
\usepackage{algorithm,algpseudocode}
\usepackage{multirow}
\usepackage{enumitem}
\usepackage{amsmath,amsfonts,latexsym,amssymb,amsbsy,amsthm}

\usepackage{footnote}
\usepackage{thmtools, thm-restate}
\usepackage{authblk}
\usepackage{soul}

\newtheorem{theorem}{Theorem}
\newtheorem{fact}[theorem]{Fact}

\newtheorem{corollary}{Corollary}

\newtheorem{lemma}{Lemma}

\newtheorem{proposition}{Proposition}

\usepackage{titlesec}
\titlespacing*{\section}{0pt}{1.1\baselineskip}{\baselineskip}

\def\prob#1#2{\mbox{Pr}_{#1}\left[ #2 \right]}

\def\expec#1#2{{\mathbb{E}}_{#1}\left[ #2 \right]}

\def\R{\mathbb{R}}

\newcommand{\beq}{\begin{equation}}
\newcommand{\eeq}{\end{equation}}

\newcommand{\alg}{\text{\sc Alg}}
\newcommand{\timebound}{\text{\sc Time}}
\newcommand{\permutations}[1]{\mathcal{S}_{#1}}

\newcommand{\closure}[1]{\text{cl}(#1)}

\usepackage{xcolor,cite,etoolbox}
\makeatletter 
\pretocmd\@bibitem{\color{black}\csname keycolor#1\endcsname}{}{\fail}
\newcommand\citecolor[1]{\@namedef{keycolor#1}{\color{blue}}}
\makeatother
\citecolor{ChekuriM97}

\usepackage{times}
\usepackage[makeroom]{cancel}
\newcommand \sge[1] {{\color{black} #1}}

\definecolor{darkslategray}{rgb}{0.18, 0.31, 0.31}
\newcommand \mw[1] {{\color{red} #1}}

\newcommand \track[1] {{\color{black} #1}}

\renewcommand\sout[1]{{\iffalse #1 \fi}}

\usepackage{pgf,tikz}
\usetikzlibrary{math,calc,arrows, shapes,positioning,decorations.pathreplacing}

\tikzstyle{vertex}=[fill=black, draw=black, shape=circle, minimum size = 5pt,inner sep=0pt]

\usepackage{authblk}

\usepackage[symbol]{footmisc}

\allowdisplaybreaks 
\begin{document}

\title{Hardness and Approximation of Submodular Minimum Linear Ordering Problems\thanks{}}

\author[1]{Majid Farhadi}
\author[2]{Swati Gupta\footnote{Corresponding Author: Swati Gupta,  Affiliation: Massachusetts Institute of Technology, 
E-mail Address: swatig@mit.edu. A part of this work was done while all the authors were at Georgia Institute of Technology.}}
\author[1]{Shengding Sun}
\author[4]{Prasad Tetali}
\author[1]{Michael C. Wigal}

\affil[1]{\small Georgia Institute of Technology \protect \\
{\small \tt \{farhadi, ssun313, wigal\}@gatech.edu}}
\affil[2]{\small Sloan School of Management, Massachusetts Institute of Technology\protect \\ {\small \tt swatig@mit.edu}}
\affil[3]{\small Cargnegie Mellon University\protect \\
{\small \tt tetali@cmu.edu}}

\maketitle

\begin{abstract}
The minimum linear ordering problem (MLOP) generalizes well-known combinatorial optimization problems such as minimum linear arrangement and minimum sum set cover. MLOP seeks to minimize an aggregated cost $f(\cdot)$ due to an ordering $\sigma$ of the items (say $[n]$), i.e., $\min_{\sigma} \sum_{i\in [n]} f(E_{i,\sigma})$, where $E_{i,\sigma}$ is the set of items mapped by $\sigma$ to indices $[i]$. Despite an extensive literature on MLOP variants and approximations for these, it was unclear whether the graphic matroid MLOP was NP-hard. We settle this question through non-trivial reductions from mininimum latency vertex cover and minimum sum vertex cover problems. We further propose a new combinatorial algorithm for approximating monotone submodular MLOP, using the theory of principal partitions. This is in contrast to the rounding algorithm by Iwata, Tetali, and Tripathi [ITT2012], using Lov\'asz extension of submodular functions. We show a $(2-\frac{1+\ell_{f}}{1+|E|})$-approximation for monotone submodular MLOP where $\ell_{f}=\frac{f(E)}{\max_{x\in E}f(\{x\})}$ satisfies $1 \leq \ell_f \leq |E|$. Our theory provides new approximation bounds for special cases of the problem, in particular a $(2-\frac{1+r(E)}{1+|E|})$-approximation for the matroid MLOP, where $f = r$ is the rank function of a matroid. We further show that \track{minimum latency vertex cover (MLVC)} is $\frac{4}{3}$-approximable, by which we also lower bound the integrality gap of its natural LP relaxation, which might be of independent interest. 
\end{abstract}

\section{Introduction}

In the Minimum Linear Ordering Problem (MLOP), \sout{formally introduced by \cite{ITT12}, }given a finite set of elements $E$, and a function over the subsets, $f : 2^E \rightarrow \R$, one seeks an ordering of the elements, i.e., a bijection $\sigma: E \rightarrow \{1, \ldots, |E|\}$, that minimizes the aggregated cost over prefixes (or equivalently suffixes) of the ordering. \track{In other words, MLOP is of the form}
$\min_{\sigma \in \permutations{E}} \sum_{i = 0}^{|E|} f( E_{i,\sigma} )$,
where $E_{i,\sigma} = \{e \in E: \sigma(e) \le i \}$, and $\permutations{E}$ is the set of permutations of $E$. This is in contrast to the classical phenomenon of minimizing a cost function over a combinatorial subset of the powerset of the elements, for example, as in the set cover problem or the minimum spanning tree problem. \sout{General MLOP seems to be intractable, thus one often assumes further properties on $f$, for example, submodularity, monotonicty, symmetry, and so on.}

\track{It is known that the MLOP is NP-hard even with additional assumptions, for example, when the set function $f(\cdot)$ is monotone and submodular, or symmetric and submodular, or supermodular (see Table \ref{tab:hardness}).  Despite a rich literature on hardness of MLOP variants, it is unclear whether the problem remains NP-hard for many structured cases, for instance when  $f(\cdot)$ is the rank function of a matroid (i.e., submodular, monotone, bounded by set size, and integral). Furthermore, much is still unknown about the related approximation guarantees.} In this work, we push the \track{envelope} of hardness and approximability for variants of submodular MLOP. \track{In particular, we show the following:}

\begin{enumerate}
\item \track{Matroid MLOP, graphic matroid MLOP, co-graphic matroid MLOP, {\it and} minimum latency vertex cover (MLVC) are NP-hard.}
\item \track{Graphic matroid is polynomially-solvable for some classes of graphs.}
\item \track{We improve the approximation factors for matroid MLOP to $2 - \frac{1+r(E)}{1+|E|}$ and  minimum latency set cover (MLSC) to $2-\theta$, where $\theta \geq \frac{2}{|E|+1}$, and it depends on the instance, by exploiting the theory of principal partitions. These results provide a refinement of the previously best-known factors for these problems \cite{ITT12}.} 
\item \track{We also show that MLVC can be approximated to $\frac{4}{3}$, improving upon the previously best approximation achieving a factor 2 \cite{ITT12}. We analyze the fractional dimension of a related poset to achieve this bound. We further lower bound the integrality gap of the natural LP relaxation for MLVC.}
\end{enumerate}

\track{Here, matroid variants of MLOP are when $f(\cdot)$ is the rank function of the corresponding matroid, and the minimum latency set (vertex) cover problems are defined on a hypergraph (graph) where the vertices must be ordered so that the sum of the maximum indices at which every hyperedge (edge) is covered must be minimized. We include precise definitions of these problems in Section \ref{sec:prelims}.}

We summarize our hardness and approximation results \track{on MLOP} in Tables \ref{tab:hardness} and \ref{tab:approximation}. The paper is structured as follows: 
\track{We give an overview of our results and techniques in Section \ref{sec:results}, discuss related work on MLOP variants in Section \ref{sec:related_work}, and present preliminaries in Section \ref{sec:prelims}. We discuss detailed proofs of our results in Sections \ref{sec:MMLOP} to \ref{sec:monotone}. We finally conclude the paper with open problems in Section \ref{sec:future}. }

\sout{We further show that some special cases of graphic MLOP are indeed polynomially solvable. }

\section{Overview of \track{results and} techniques}\label{sec:results}

\track{We next present an overview of our results and techniques.}

\paragraph{1. Hardness of matroid MLOP:}

\sout{ We show that the optimal solutions for matroid MLOP for uniform matroids have a specific {\it flat-like} structure.} 
\sout{  Detecting this flat-like structure is equivalent to checking if a matroid is uniform,}

\track{We first show in Section \ref{sec:MMLOP} the NP-hardness of matroid MLOP, by observing the fact that a uniform matroid on a ground set $E ~(|E| = n)$ with rank $k$ has the unique property (up to isomorphism) of having ${n \choose k}$ independent sets of size $k$. We will show that {\it any} optimal matroid MLOP solution can detect this, thereby reducing the ``uniform matroid isomorphism"
 problem (known to be NP-hard \cite{OW02}) to matroid MLOP. }

\begin{restatable}{theorem}{MMLOPTHM}
\label{thm:MMLOP}
Matroid MLOP is NP-hard.
\end{restatable}
Furthermore, we show that matroid MLOP in decision form on a family of matroids shares the same complexity class with matroid MLOP in decision form on the matroidal {\it dual} family. This observation will be useful for upcoming results.

\paragraph{2. Hardness of graphic matroid MLOP:} Next, in Section~\ref{sec:Graphic_MLOP}, we further restrict $f(\cdot)$ to the special case of the rank function of any {\it graphic} matroid. Tutte \cite{tutte1959matroids} gave a complete minor-free characterization for graphic matroids. In particular, graphic matroids are {\it regular}, i.e., representable using a totally unimodular matrix, and in particular, do not contain a rank-2 uniform matroid over 4 elements as a minor (e.g., see \cite{O06}). Therefore for graphic matroids, the reduction from uniform matroid isomorphism does not suffice. We show that it is NP-hard using a series of reductions beginning at \track{the minimum sum vertex cover (MSVC) on simple graphs $G$, which we show reduces to the minimum latency vertex cover (MLVC) on the complement graph $\bar{G}$, which we show finally reduces to the graphic matroid MLOP.} 

\begin{restatable}{theorem}{GMMLOPTHM}
\label{thm:GMMLOP}
Graphic matroid MLOP is NP-hard.
\end{restatable}

To reduce MLVC to the graphic matroid MLOP with graph $G = (V,E)$, we first create an auxiliary graph $H$ by adding a new vertex $z$ to $V$, and a weighted star graph $T$ centered at $z$, connected to each vertex in $V$. We choose the edge weights for $T$ in such a way that they each induce a distinct flat in any optimal ordering for weighted graphic MLOP. This implies solving (weighted) graphic matroid MLOP for $H$ is equivalent to solving MLVC on $G$. As we can keep the magnitude of the weights controlled, this allows us to reduce MLVC to graphic matroid MLOP, thereby showing hardness of the latter. \track{As a by product, we also show that MLVC and the co-graphic matroid MLOP are NP-hard, which was not known before our work.}





\paragraph{3. Improved approximation of MLSC:} 

\track{We present two different approximation algorithms for minimum latency set cover problem, both of which refine the best-known constant 2-approximation with improvements in different settings, and use different techniques. 
The MLSC can be modeled as a covering problem on a hypergraph (see Section \ref{sec:related_work} for details), and in Section \ref{sec:MLVC}, we present our first randomized approximation algorithm based on scheduling theory, whose approximation factor depends on rank of the hypergraph, i.e. the maximum cardinality of its hyperedges (or in other words, the maximum cardinality of the candidate sets) (Theorem \ref{thm:MLVC}). Our second approach simply applies the approximation for monotone submodular MLOP to MLSC, as the latter is a special case of the former (Corollary \ref{cor:MLSC_MSMLOP}). Both the resultant approximation factors depend on the properties of the instance, and none of them dominate the other on all instances.} They both improve on the previous best-known approximation bound for MLSC of $2$, using a reduction to the single machine scheduling problem with precedence constraints \cite{chekuri1999precedence,hall1997scheduling,margot2003decompositions}.


\begin{restatable}{theorem}{MLVCTHM}
\label{thm:MLVC}
    There is a randomized polynomial time algorithm that approximates MLSC within factor $2-\frac{2}{1+\ell}$, \track{where $\ell$ is the maximum cardinality among all hyperedges of $H$}.
\end{restatable}

The idea for achieving \track{our improved} approximation bound for MLSC is to exploit the structural complexity of the precedence constraints (corresponding to a poset) for the subsequent scheduling instance. Bounding the fractional dimension of this poset by $1+\ell$ allows us to utilize a state-of-the-art scheduling algorithm by \track{Amb\"{u}hl et al.} \cite{AMMS11} to approximate the objective by a factor of $2-\frac{2}{1+\ell}$. \track{For the special case where the input is a graph, this algorithm gives a factor $\frac{4}{3}$ \-approximation for MLVC. 

\begin{restatable}{corollary}{MLVCCOR}
\label{cor:MLVC}
    There exists a randomized polynomial time factor $\frac{4}{3}$-approximation algorithm for MLVC.
\end{restatable}

To the best of our knowledge, this is the current best approximation factor for MLVC.} \track{For $\ell$-uniform regular hypergraphs, i.e., where each hyperedge has size $\ell$ and each vertex is contained in the same amount of hyperedges, we show that a simple LP relaxation also achieves the $2-\frac{2}{1+\ell}$ approximation factor. In particular, the LP relaxation gives a factor $\frac{4}{3}$-approximation algorithm for MLVC on regular graphs. From this result, we raise the question whether the LP relaxation for MLSC on $\ell$-uniform hypergraphs has the same $2-\frac{2}{1+\ell}$ approximation factor. Indeed, a better approximation factor seems unlikely, as we observe a lower bound of $2-\frac{2}{1+\ell}$ on integrality gap of the LP relaxation for MLSC on $\ell$-uniform hypergraphs, matching our current approximation result. }


\track{In Section \ref{subsec:MLSC_MSMLOP}, we discuss the use of principal partitions to obtain an approximation for MLSC as a special case.}

\begin{restatable}{corollary}{MLSCMSMLOP}
\label{cor:MLSC_MSMLOP}
    There is a deterministic factor $(2-\frac{\Delta+|E|}{\Delta(1+|V|)})$-approximation algorithm for MLSC, where $\Delta$ is the maximum degree of hypergraph $H = (V,E)$. 
\end{restatable}

\track{Note that $\Delta=\max_{v\in V}|\{e\in E:v\in e\}|$. Together Theorem \ref{thm:MLVC} and Corollary \ref{cor:MLSC_MSMLOP} imply that MLSC can be approximated within factor $2-\theta$, where $\theta=\max\{\frac{2}{1+\ell},\frac{\Delta+|E|}{\Delta(1+|V|)}\}$. Note $\theta$ can be very small, for example, for $\ell$-uniform hypergraphs where $\ell$ is large. However, since $\theta \geq \frac{2}{n+1}$, we get a slight improvement over 2.}

\paragraph{4. Polynomially solvable instances of matroid MLOP:} In Section \ref{subsec:cactus_graphs}, we propose a novel characterization of matroid MLOP, wherein one can search through bases and permutations of bases, rather than permutations of the ground set. In particular, whenever the number of bases of a matroid is small (polynomial in $|E|$) and the rank of the matroid is also small (constant), we show that matroid MLOP becomes polynomial time solvable.

\begin{restatable}{theorem}{SmallBasisMLOP}
\label{thm:small_basis_MLOP}
    Let ${\cal X}$ be a family of matroids \track{such that for all $M = (E,r_M) \in {\cal X}$ with $|E| = m$, the number of bases of $M$ is $|\mathcal{B}(M)| \in O(g(m))$, and the rank of $M$ is $r_M(m) \in O(h(m))$, for some }$g, h : \mathbb{Z}_+ \to \mathbb{Z}_+$. \sout{such that $g(m)$ bounds the size of ${\cal B}(M)$, and let $h : \mathbb{Z}_+ \to \mathbb{Z}_+$ such that $h(m)$ bounds the size of $r(E)$ {[\color{blue} what is the size of $r(E)$? isn't that always just at most $m$? why not simply say that $|\mathcal{B}(M)| \in O(g(m))$ and $r_M(E) \in O(h(m))$]}.} Then, \track{every matroid MLOP instance in} ${\cal X}$ can be solved in time $O(g(m) \cdot poly(m,g(m)) \cdot (h(m))!)$ In particular, if $g$ is polynomial in $m$ and $h$ is bounded by a constant, then matroid MLOP for ${\cal X}$ is in $P$.
\end{restatable}

For the special case of graphic matroid MLOP on cactus graphs\footnote{A graph $G$ is a cactus graph if every maximal 2-connected subgraph of $G$ is a cycle or an edge. }, we show that an optimal MLOP ordering can be found by fixing any spanning tree of the cactus graph. To find an ordering of the edges of the spanning tree, we show that a greedy ordering on the cycles of the cactus graph suffices (even though the size of the basis may not be logarithmic in size with respect to the ground set).

\begin{restatable}{theorem}{CacMLOPTHM}
\label{thm:CacMLOP}
Given a simple cactus graph $G$, there is a polynomial time algorithm that solves graphic matroid MLOP on $G$.
\end{restatable}

Furthermore, in Section \ref{subsec:MLVC_poly}, we show how if a graph is regular, then the optimal objective values for MLA, MSVC, and MLVC are all related by linear shifts \track{in the objective parameterized by the number and degree of the vertices}. As many instances of regular graphs have polynomial time algorithms (e.g. see \cite{diaz2002survey,petit2013addenda,bezrukov1999edge,lai1999survey}) this leads to many new polynomial time algorithms for MSVC and MLVC for many instances of regular graphs.

\paragraph{5. Improved approximation for monotone submodular MLOP:} \track{For} monotone submodular MLOP, Iwata, Tetali, and Tripathi \cite{ITT12} provided a $(2-\frac{2}{|E|+1})$-approximation algorithm for monotone submodular MLOP based on Lov{\'a}sz extension in 2012. Another natural approach is to use the theory of principal partitions induced by a given submodular function \cite{PrincipalPartitionKK,PrincipalPartitionF}. The principal partition of a ground set $E$ of a monotone submodular function is a chain of subsets $\mathcal{C} = \emptyset \subseteq S_1 \subseteq \hdots \subseteq S_k = E$, such that each $S_i$ is the unique maximal minimizer of $f(S) - \lambda_i |S|$ for some $\lambda_i \in \mathbb{R}_+$. As early as 1992, Pisaruk considered completing the chain $\mathcal{C}$ randomly to add subsets of missing cardinality  (~\cite{pisaruk1992boundaries}, c.f. \cite{fokkink2019submodular}). Later in 2019, Fokkink et al. \cite{fokkink2019submodular} considered the same algorithm for the submodular search problem, which includes monotone submodular MLOP as a special case. They showed that this algorithm has an approximation ratio based on the total curvature\footnote{The total curvature of a set function $f$ is defined to be $\max_{x\in E}\frac{f(\{x\})+f(E-x)-f(E)}{f(\{x\})}$, e.g., see \cite{fokkink2019submodular}.} of the submodular function, and is always at most 2.

It was not known how these two results compare, as they use very different techniques. We show that the algorithm based on principal partitions always has better approximation guarantee than the $(2-\frac{2}{|E|+1})$ bound of Lov{\'a}sz extension relaxation proven in \cite{ITT12}.

\begin{restatable}{theorem}{MSMLOP}
\label{thm:MSMLOP}
Let $f:2^E \to \R$ be a non-trivial, normalized and monotone submodular function. There exists a factor $(2-\frac{1+\ell_{f}}{1+|E|})$-approximation algorithm to MLOP with $f(\cdot)$ in polynomial time, where $\ell_{f}=\frac{f(E)}{\max_{x\in E}f(\{x\})}$. 
\end{restatable}

As $\ell_f$ is bounded \track{below by 1}, the above result is a refinement of the previous $(2 - \frac{2}{|E| + 1})$-approximation \cite{ITT12}. \track{Our result is also independent from the analysis in Fokkink et al. \cite{fokkink2019submodular} using total curvature, and leads to nice approximation bounds for some classes of matroids where $\ell_f$ is large.} 
For example, for graphic matroid MLOP on \track{connected} graphs of \track{bounded} maximum degree \track{$\Delta$ with $\Delta > 1$}, we obtain a $(2 - \frac{2}{\Delta})$-approximation \track{asymptotically}. \track{This constant factor improvement from 2 cannot be obtained using either the Lov{\'a}sz extension bound in \cite{ITT12} or the total curvature bound in \cite{fokkink2019submodular}. }

Our results have led to multiple open questions which may be of independent interest, and are discussed in Section~\ref{sec:future}.

\section{Related work}\label{sec:related_work}

MLOP was formally introduced by Iwata et al. \cite{ITT12}, generalizing many \track{well-known} combinatorial optimization problems. \track{In this section, we describe related work in combinatorial optimization that can be viewed as different instances of MLOP.} Some of these \track{MLOP variants (e.g., minimum latency set cover)} will be utilized in our proof that the graphic MLOP is NP-hard, as \track{depicted} in Figure \ref{fig:summary}.
\sout{such as the classical minimum linear arrangement (MLA) problem. }

\begin{figure}[t]
\centering
\begin{tikzpicture}[->,>=stealth',shorten >=1pt,auto,thick]

        \node[ellipse,draw,minimum size=1.5cm] (MLOP) {MLOP};
        
        \node[ellipse,draw,text width=2cm,align=center] (supMLOP) [below=1cm of MLOP] {supermodular MLOP};
        
        \node[ellipse,draw,text width=1cm,minimum size=1cm,align=center] (MSSC) [left=0.5cm of supMLOP] {MSSC};
        
        \node[ellipse,draw,text width=1cm,minimum size=1cm,align=center] (MSVC) [left=0.5cm of MSSC] {MSVC};
        
        \node[ellipse,draw,text width=1.5cm,minimum size=1cm,align=center] (MIR) [left=0.5cm of MLOP] {MIR/GMSSC};
        
        \node[ellipse,draw,text width=1cm,minimum size=1cm,align=center] (MLSC) [left=0.5cm of MIR] {MLSC};
        
        \node[ellipse,draw,text width=1cm,minimum size=1cm,align=center] (MLVC) [left=0.5cm of MLSC] {MLVC};
        
        \node[ellipse,draw,text width=2cm,align=center] (subMLOP) [above=1cm of MLOP] {submodular MLOP};
        
        
        \node[ellipse,draw,text width=1cm,minimum size=1cm,align=center] (MSMLOP) [left=0.5cm of subMLOP] {Mon.\ Sub.\ MLOP};
        
        \node[ellipse,draw,text width=1cm,minimum size=1cm,align=center] (MMLOP) [left=0.5cm of MSMLOP] {Matroid MLOP};
        
        \node[ellipse,draw,text width=1cm,minimum size=1cm,align=center] (GMLOP) [left=0.5cm of MMLOP] {Graphic MLOP};
        
        
        \path (GMLOP) edge (MMLOP)
            (MMLOP) edge (MSMLOP)
            (MSMLOP) edge (subMLOP)
            (subMLOP) edge (MLOP);
            
        \path (MLVC) edge (MLSC)
            (MLSC) edge (MIR)
            (MIR) edge (MLOP);
        \path (MSVC) edge (MSSC)
            (MSSC) edge (supMLOP)
            (supMLOP) edge (MLOP);
        
        
        \path (MSSC) edge (MIR);
        
        \path (MLSC) edge (MSMLOP);
        
        
        \path (MSVC) edge[<->,dashed] (MLVC)
            (MLVC) edge[->,dashed] (GMLOP);
		
\end{tikzpicture}
\caption{\small Overview of related problems. A solid arrow from problem $A$ to $B$ indicates that $B$ generalizes $A$. A dashed arrow from problem $A$ to $B$ denotes that computation of $A$ can be polynomially reduced to computation of $B$, using our gadgets.}
\label{fig:summary}
\end{figure}

\paragraph{Minimum linear arrangement (MLA)}

Motivated by applications in coding theory, Harper \cite{harper1964optimal} introduced minimum linear arrangement (MLA) in 1964, which seeks to find an arrangement of the vertices of a given graph $G = (V,E)$ such that the total ``stretch'' of each edge is minimized, i.e., MLA on a graph $G = (V,E)$ is the following,
\begin{align*}
    \min_{\pi \in \permutations{V}}\sum_{(u,v) \in E} |\pi(u) - \pi(v)|.
\end{align*}
Note that any permutation $\pi \in \permutations{V}$ naturally induces a chain on $V$ with prefix sets $V_{i,\pi} = \{v \in V : \pi(v) \le i\}$. Let $\phi$ be the cut function of the graph, i.e., for all $S \subseteq V$, $\phi(S)$ is the number of edges with exactly one end in $S$. Note then for any permutation $\pi \in \permutations{V}$ if an edge $(u,v) \in E$ is stretched to a value $k = |\pi(u) - \pi(v)|$, it must cross the cut of exactly $k$ prefix sets in the chain $V_{0,\pi} \subsetneq V_{1,\pi} \subsetneq \cdots \subsetneq V_{n-1,\pi} \subsetneq V_{n,\pi}$ where $|V(G)| = n$. Thus, MLA on a graph $G = (V,E)$ is equivalent to \vspace{-0.2cm}
\begin{align*}
    \min_{\pi \in \permutations{V}}\sum_{i = 0}^{n} \phi(V_{i,\pi}),
\end{align*}
which is an instance of MLOP with $\phi$ being a symmetric submodular function. Solving MLA for specific instances of graphs has received considerable attention due to its many applications, see surveys \cite{diaz2002survey,petit2013addenda,bezrukov1999edge,lai1999survey}. While MLA is polynomial time solvable for some classes of graphs, for example trees \cite{MLATREE, chung1984optimal},  its \track{decision form} has been known to be NP-complete since 1974 \cite{garey1974some}.\sout{ In practice, MLA is a notorious problem with the best known approximation bound of $O(\sqrt{\log n} \log \log n)$}
\track{The best known approximation bound for MLA is $O(\sqrt{\log n} \log \log n)$} (\cite{FL07}, \cite{charikar2010}). \sout{ For hardness of approximation, see \cite{ambuhl2011inapproximability}. } \track{ Under the \track{exponential time hypothesis\cite{impagliazzo2001complexity}} that there does not exist a \sge{randomized algorithm} to solve SAT in time $2^{n^{\varepsilon}}$ where $n$ is the instance size and $\varepsilon > 0$ is arbitrarily small, it is also known that MLA is inapproximable to some constant \cite{ambuhl2011inapproximability}}.  \sout{This illustrates that while the minimization of a given cost function may be easy (i.e., finding a minimum cut in a graph), the corresponding MLOP with the same function can be a difficult problem. }

\begin{table}[t]
\footnotesize
\centering
\caption{Previously known results and our results on NP-hardness of different MLOP variants} \label{tab:hardness}
\begin{tabular}{ |c|c|c|c| }
\hline 
 MLOP Class & Problem & Hardness & Source\\\hline\hline 
 general  & multiple intents ranking (MIR) & NP-hard & Azar et al. \cite{AGY09} \\\hline
 \multirow{3}{*}{monotone supermodular} & minimum sum set cover (MSSC) & NP-hard 
 & Feige et al. \cite{FLT04}\\ \cline{2-4} 
   & minimum sum vertex cover (MSVC) & NP-hard 
 & Feige et al. \cite{FLT04}\\\hline
 \multirow{3}{*}{monotone submodular}  & matroid MLOP & {\bf NP-hard}  & \textbf{Theorem~\ref{thm:MMLOP}} \\  \cline{2-4} 
   & graphic matroid MLOP & {\bf NP-hard}  & {\bf Theorem~\ref{thm:GMMLOP}} \\\cline{2-4} 
   & co-graphic matroid MLOP & {\bf NP-hard} &{\bf Corollary~\ref{cor:CGMMLOP}}\\\cline{2-4} 
  & graphic matroid MLOP for cactus graphs & {\bf P}  & {\bf Theorem \ref{thm:CacMLOP}}\\ \cline{2-4}
  & minimum latency set cover (MLSC) & NP-hard & Hassin and Levin \cite{HL05} \\ \cline{2-4} 
   & minimum latency vertex cover (MLVC) & {\bf NP-hard}  & {\bf Theorem \ref{theo:MLVC_MSVC_equivalence}}\\ \hline 
submodular & sum cut (SUMCUT) & NP-hard &  \cite{diaz1991minsumcut,lin1994profile}\\ \hline
symmetric submodular & minimum linear arrangement (MLA) & NP-hard  & \cite{garey1974some, ES75}\\ \hline
\end{tabular}
\end{table}



\paragraph{Minimum latency set cover (MLSC)}

\track{MLSC was introduced by Hassin and Levin \cite{HL05} with motivations from problems in job scheduling, and they provided an $e$-approximation. The best known approximation constant for MLSC is 2 \cite{HL05,AGY09}. Later in our work, we show that MLSC can be viewed as an instance of monotone submodular MLOP, for which Iwata, Tetali and Tripathi \cite{ITT12} gave a factor $(2-\frac{2}{|E|+1})$ approximation algorithm using the Lov{\'a}sz extension. We give a more refined approximation algorithm for monotone submodular MLOP using principal partitions, which applies to MLSC as well.} \sout{Monotone submodular MLOP is generalized by the submodular search problem, which was studied in \cite{fokkink2019submodular}.}
\begin{table}[t]
\footnotesize
\centering
\caption{\track{Summary of approximation factors known for MLOP variants. For MLSC, $\theta=\max(\frac{2}{1+\ell},\frac{\Delta+|E|}{\Delta(1+|V|)})$, where $\ell$ is the maximum cardinality of hyperedges, and $\Delta$ is the maximum degree in the graph. For graphic MLOP, we assume that the graph is connected.}} \label{tab:approximation}
\begin{tabular}{ |c|c|c| }
\hline 
Problem & Approximation & Source\\\hline
\hline & &\\[-8pt]
  Matroid MLOP & {\bf \footnotesize  $\mathbf{2-\frac{1+r(E)}{1+|E|}}$} & \textbf{Corollary 
\ref{cor:matroid_MLOP_approximation}}\\[3pt]
\hline & &\\[-8pt]
Graphic MLOP\footnote{\track{Assuming graph is connected.}} & {\bf \footnotesize  $\mathbf{2-\frac{|V(G)|}{1 + |E(G)|}}$} & \textbf{Corollary \ref{cor:matroid_MLOP_approximation}}\\ [3pt]
\hline & &\\[-8pt]
Monotone Submodular MLOP & {\bf \footnotesize  $\mathbf{2-\frac{1+\ell_f}{1+|E|}}$} & \textbf{Theorem \ref{thm:MSMLOP}}\\[3pt]
\hline & &\\[-8pt]
  MLA & {\footnotesize $O(\sqrt{\log n} \log \log n)$} & {\track{Feige and Lee} \cite{FL07}}\\[1pt]  
\hline & &\\[-8pt]
MLVC & {\bf \footnotesize \ $\mathbf{\frac{4}{3}}$} & \textbf{Theorem \ref{thm:MLVC}}\\[3pt] 
\hline & &\\[-8pt] 
MLSC & {\bf \footnotesize  $\mathbf{2-\theta}$} & \textbf{ Theorem \ref{thm:MLVC} \& Corollary \ref{cor:MLSC_MSMLOP}}\\
\hline & &\\[-8pt] 
MIR & {\footnotesize $4.642$} & \track{{Bansal et al. \cite{BBFT21}}}\\ 
\hline & &\\[-8pt] 
SUMCUT & {\footnotesize $O(\log n)$} & {\track{Rao and Richa} \cite{rao2005new}}\\
\hline & &\\[-8pt] 
Supermodular MLOP & {\footnotesize $4$} & \track{Iwata, Tetali, Tripathi} \cite{ITT12}\\
\hline & &\\[-8pt] 
MSSC & {\footnotesize $4$} & \track{Feige, Lov{\'a}sz, Tetali \cite{FLT04}}\\
\hline & &\\[-8pt]
 MSVC & {\footnotesize $\frac{16}{9}$} & \track{Bansal et al. \cite{BBFT21}}\\[2pt] \hline
\end{tabular}
\end{table}

\paragraph{Minimum sum set cover}

Minimum sum set cover (MSSC) was introduced by Feige, Lov{\'a}sz, and Tetali \cite{FLT04}, \track{who also presented} a greedy algorithm that provides a $4$-approximate solution to MSSC, and showed it is NP-hard to do better. Later, Iwata, Tetali, and Tripathi \cite{ITT12} showed \track{that MSSC is an instance of supermodular MLOP, and the greedy algorithm for MSSC can be generalized to approximate supermodular MLOP within factor 4.}

\track{MSSC can be formulated as follows, using the notation of hypergraphs: given a hypergraph $H=(V(H),E(H))$, MSSC seeks to find a permutation of vertices that minimizes the total costs of all hyperedges, where the cost of each hyperedge is the minimum of its vertex labels, i.e.,

\begin{align*}
    \min_{\pi \in \permutations{V(H)}} \sum_{e \in E(H)} \min_{v\in e}\pi(v). 
\end{align*} }

\track{The special case when $H$ is a graph is the well-known minimum sum vertex cover (MSVC). Independent from MSSC, MSVC was introduced earlier by} Burer and Monteiro \cite{BM01} as a heuristic in solving semidefinite relaxation of the Max-Cut problem. Feige, Lov{\'a}sz, and Tetali \cite{FLT04} later showed MSVC has a 2-approximation based on linear programming rounding, and \track{also showed that it is NP-hard to approximate for an unknown constant $\epsilon$, where $1<\epsilon<2$}. Later, Barenholz, Feige and Peleg \cite{BFP06} improved to a 1.9946-approximation, and recently, Bansal et al. \cite{BBFT21} gave a $\frac{16}{9}$-approximation for MSVC. \track{The best possible approximation constant for MSVC is still unknown.} For the special case of regular graphs, Feige, Lov{\'a}sz, and Tetali \cite{FLT04} gave a $\frac{4}{3}$-approximation. \track{This approximation guarantee for regular graphs was later improved by Stankovi{\'c} \cite{stankovic2022some} to 1.225.}

\track{These problems concern supermodular functions, but we only consider submodular functions in this work.}

\paragraph{Other variants of MLOP} 

Another variant of MLOP is called the multiple intents ranking (MIR), and has been studied in  \cite{AGY09,BGK10,SW11,ISV14,BBFT21}. Azar, \track{Gamzu, and Yin} \cite{AGY09} gave a 2-approximation for MIR, \sge{for the case when the weight vector for each hyperedge is monotonically non-decreasing. This variant of MIR} includes MLSC as a special case. 
These problems have found a broad spectrum of applications in query results diversification \cite{tsaparas2011selecting}, motion planning for robots \cite{LCS16}, cost-minimizing search \cite{FLV19}, and optimal scheduling \cite{happach2020min}, among others.

\track{Another example of an instance of submodular MLOP is the sum cut problem (SUMCUT). The problem was independently introduced by D{\'\i}az et al. \cite{diaz1991minsumcut} and also Yixun and Jinjiang \cite{lin1994profile} to study circuit layouts. SUMCUT is NP-complete \cite{diaz1991minsumcut,lin1994profile} and Rao and Richa \cite{rao2005new} gave a $O(\log n)$-approximation algorithm for SUMCUT using a divide-and-conquer approach.} 

Recently, Happach, Hellerstein, and Lidbetter \cite{HHL20} viewed MLOP under the umbrella of minimum sum ordering/permutation problem, and generalized results of Feige, Lov{\'a}sz, and Tetali \cite{FLT04}.

\section{Preliminaries}\label{sec:prelims}

\track{We now present notation and background useful for parsing this work. We refer an interested reader to \cite{Schrijver03} for further reading.} 

\paragraph{1. Submodular set functions}

 For a set of elements $S$ and elements $x\notin S,y\in S$, we use $S+x,S-y$ to denote $S\cup \{x\},S\setminus \{y\}$ respectively. Let $f:2^E\to\R$ be a set function. We say $f$ is \textit{submodular} if for all $S,T\subseteq E$, $f(S)+f(T)\ge f(S\cup T)+f(S\cap T)$. An equivalent definition is $f(S+e)-f(S)\ge f(T+e)-f(T)$ for all $S\subseteq T$ and $e\notin T$. This property is sometimes called the diminishing return property. A set function $f$ is \textit{supermodular} if $-f$ is submodular, and is \textit{symmetric} if $f(S)=f(E\setminus S)$ for all $S\subseteq E$ and is \textit{monotone} if $f(S)\le f(T)$ for all $S\subseteq T\subseteq E$. We say $f$ is \textit{normalized} if $f(\emptyset)=0$. A normalized monotone submodular function $f$ is \textit{non-trivial} if $f(E)\ne 0$, i.e., $f$ is not identically zero. For a normalized non-rivial monotone submodular function $f$, we define the \textit{steepness} of $f$ as $\kappa_{f}=\max_{x\in E}f(\{x\})$, which is the maximum function value of any singleton. We further define the \textit{linearity} of $f$ as $\ell_{f}=\frac{f(E)}{\kappa_{f}}$. For a normalized symmetric submodular function $f$, and $s,t \in E$, an \textit{$s - t$ cut} is a subset $X \subseteq E$ such that $s \in X$, $t \not \in X$ or $s \not \in X$, $t \in X$. The cut is said to have value $f(X) = f(E \setminus X)$.

For a finite set $E$ with size $m>0$, we define $\permutations{E}$ to be the set of all bijective functions $\sigma : E \rightarrow \{1, \ldots, |E|\}$. For every $\sigma \in \permutations{E}$, we define the prefix sets $E_{i,\sigma} = \{e \in E : \sigma(e) \le i \}$. Note $|E_{i,\sigma}|=i$ for all $1\le i\le m$ and $\emptyset \subsetneq E_{1,\sigma} \subsetneq E_{2,\sigma} \subsetneq \cdots \subsetneq E_{m,\sigma} = E$. 

\paragraph{2. Matroids}

A {\it rank} function \track{(for a matroid)} is an \track{integer-valued nonnegative} monotone submodular function \sout{$r: 2^{E} \rightarrow \mathbb{Z}_+$}\track{$r: 2^{E} \rightarrow \mathbb{Z}_{\ge 0}$}, such that $r(A) \leq |A|$ for all $A\subseteq E$. A pair $M = (E,r)$ where $r$ is a rank function on $E$ is a \textit{matroid}. There \track{are} multiple \sout{cryptomorphic}\track{equivalent} definitions for matroids and we refer to \cite{O06} for other equivalent definitions and basic theory. Note if $r$ is a rank function \track{of a non-trivial matroid (where $r(E)>0$)} then $\kappa_{r}=1$ and $\ell_{r}= r(E)$. The set $E$ is the \textit{ground set} of the matroid $M$, which we also denote $E(M)$. A set $I \subseteq E$ is \textit{independent} if $r(I) = |I|$, and is \textit{dependent} otherwise. A maximal independent set is a \textit{basis}, and the set of all bases of $M$ is denoted as ${\cal B}(M)$. A \textit{circuit} of $M$ is a minimally dependent set. Let $e,e' \in E(M)$ for some matroid $M$, then $e$ is a \textit{loop} if $\{e\}$ is a circuit and $e$ and $e'$ are \textit{parallel} if $\{e,e'\}$ is a circuit. Let $B$ be a basis for $M$ and note for all $e \in E \setminus B$, $B + e$ is a dependent set. It is well known that $B + e$ contains a unique circuit, called the \textit{fundamental circuit} of $e$ with respect to $B$, which will we denote $C(B,e)$. A \textit{flat} of a matroid is a subset $X \subset E$ that is maximal with respect to its rank. The \textit{closure} of a set $S \subseteq E$, is $\closure{S} = \{x \in E : r(S \cup \{x\}) = r(S)\}$.

A matroid is \textit{uniform} of rank $k$ if its bases consists of all subsets of size $k$. We denote a uniform matroid of size $m$ and of rank $k$ as $U_k^m$. If the independent sets of a matroid $M$ is the family of acyclic sets of a graph $G$, then $M$ is a \textit{graphic matroid}, which we denote $M = M[G]$. If $M$ is a matroid with rank function $r$, then its \textit{corank} function is the following, $r^*(X) = |X| - r(M) + r(E \setminus X)$. It is well known, see \cite{O06}, that $r^*$ is also a rank function for a matroid $M^*$ on $E(M)$, and we let $M^*$ denote the \textit{dual matroid} of $M$. An element is a \textit{coloop} if it is a loop in the dual matroid. A \textit{cographic matroid} is a matroid whose dual matroid is graphic.

For a positive integer $m$, let $[m] = \{1,2, \ldots, m\}$. Given a $k\times m$ matrix $A$ with integer entries, the vector matroid of $A$, denoted by $M[A]$, is defined as follows: the ground set is $[m]$, and the rank function of $J\subset [m]$ is the \track{(matrix)} rank of $k\times |J|$ submatrix $A_J$, which is obtained from $A$ by deleting columns whose index is not in $J$.

\track{Given a matroid $M = (E, r)$, matroid MLOP solves  $\min_{\sigma \in \permutations{E}} \sum_{i = 0}^{|E|} r( E_{i,\sigma} )$,
where $E_{i,\sigma} = \{e \in E: \sigma(e) \le i \}$, and $\permutations{E}$ is the set of permutations of $E$.}

\paragraph{3. Graphs, \track{hypergraphs and partial orders}}

A \textit{graph} $G$ over a set of vertices $V(G)$, can be defined by a multiset of edges $E(G) \subseteq V \times V$. We allow graphs to have \track{multiedges and }loops, and a graph is \textit{simple} if it does not have multiedges or loops. A graph is a \textit{clique} if every pair of distinct vertices has a single edge joining them. The clique or complete graph on $n$ vertices is denoted $K_n$. The \textit{complement} of a simple graph $G$, denoted $\overline{G}$, is the graph \track{where $V(\overline{G}) = V(G)$ and for all distint $u,v \in V(G)$,  $(u,v) \in E(\overline{G})$ if and only if $(u,v) \not \in E(G)$.} A \textit{block} of a graph $G$ is a maximal connected subgraph, without a cut vertex. Note that \sout{a }the blocks of $G$ \sout{is}\track{are} either an edge or a 2-connected subgraph. It is well known that every pair of distinct blocks are either disjoint or intersect at a cut vertex of $G$. A \textit{cactus graph} is graph $G$ in which every block of $G$ is an edge or a circuit. \track{A hypergraph $H = (V, E)$ is a generalization of graphs that allows each edge $e \in E$ to be a subset of vertices $V$, where each such subset is referred to as a hyperedge. A graph is a special case of a hypergraph where all hyperedges have size 2.}

\track{We now define the minimum latency set cover (MLSC). Given a hypergraph, MLSC asks to find a permutation on its vertices that minimizes the aggregated cost of the hyperedges, where the cost of an hyperedge is the maximum label of its vertices.
\begin{align*}
    \min_{\pi \in \permutations{V(H)}} \sum_{e \in E(H)} \max_{e \in E(H)} \pi(v)
\end{align*}
The minimum latency vertex cover (MLVC) is an instance of MLSC where the input is restricted to being a graph.}

\track{A \textit{partially ordered set} or \textit{poset} is a pair $(P,<_P)$ where $P$ is a set and $<_P$ is an antisymmetric and transitive relation on $P$, i.e., such that for all distinct $x,y,z\in P$, we have that
\begin{enumerate}
    \item if $x <_P y$ and $y <_P z$ then $x <_P z$, and 
    \item if $x <_P y$ then $y \not <_P x$.  
\end{enumerate}
For any $x,y \in P$ we say $x$ and $y$ are \textit{comparable} if $x <_P y$, $y <_P x$ or $x = y$. A \textit{chain} is a subset $S \subseteq P$ of pairwise comparable elements of $(P,<_P)$. A partial order $(P,<_P)$ is a \textit{total order} if $P$ is a chain. A poset $(P,\prec_P)$ is an \textit{extension} of a poset $(P,<_P)$ if for all $x,y \in P$, if $x <_P y$ implies $x \prec_P y$. An extension $(P,\prec_P)$ is \textit{linear} if $(P,\prec_P)$ is a total order.
}

\paragraph{4. Principal partitions}

We refer the readers to \cite{PrincipalPartitionF,PrincipalPartitionN} for the general theory on principal partitions. Here we state some properties of principal partitions on monotone submodular functions. 

\begin{theorem}[\cite{PrincipalPartitionN}]\label{thm:pp}
Let $f$ be a monotone submodular function such that $f(A)=0$ if and only if $A=\emptyset$. Then there exist positive integer $s\ge 1$ and nested sets $\emptyset=\Pi_0\subsetneq \cdots \subsetneq\Pi_s=E$, called principal partitions of $f$, as well as real numbers $\lambda_0<\lambda_1< \cdots <\lambda_{s+1}$, called critical values, such that for all $0\le i\le s$, $\Pi_i$ is the unique maximal optimal solution to $\min_{X\subseteq E} f(X)-\lambda |X|$, for all $\lambda\in(\lambda_i,\lambda_{i+1})$. Furthermore,  $\{\Pi_i\}_{0\le i\le s}$ as well as $\{\lambda_i\}_{1\le i\le s}$ can be computed in polynomial time. 
\end{theorem}
Some authors refer to $\{\Pi_i\}_{0\le i\le s}$ as the principal sequence of partitions, and/or use the minimal (which is also unique) instead of maximal optimal solution. Note that the principal partitions minimize the function value among subsets of the same size.

\section{Matroid MLOP is as hard as uniform matroid isomorphism}\label{sec:MMLOP}

Before we show NP-hardness of graphic matroid MLOP, we will first show in this section that the more general case that matroid MLOP is indeed NP-hard using a reduction to \track{the} uniform matroid isomorphism problem. Although the uniform matroid is one of the simplest matroids, it turns out that determining whether a given matroid is uniform is NP-hard \cite{OW02}. Formally, the uniform matroid isomorphism problem is the following: 
 
 \begin{center}
 {\it Given a $k \times m$ matrix $A$ with integer entries, is $M[A]$ isomorphic to $U^k_m$?}
 \end{center}
where $M[A]$ denotes the vector matroid of $A$. 

\track{In the} following lemma \track{we argue} that the \track{optimal} matroid MLOP value is unique for each uniform matroid. This provides a reduction to the uniform matroid isomorphism problem.

\begin{lemma}\label{lem:mlop_uniform}
Let $M =(E, r)$ be a matroid, \track{of size $|E| = m$}, and rank at most \track{$k \ge 1$}. We have $$\min_{\sigma \in \permutations{E}}\sum_{i = 1}^m r(E_{i,\sigma}) = {k + 1 \choose 2} + k(m-k),$$ if and only if $M$ is isomorphic to \track{the uniform matroid} $U_k^m$. 
\end{lemma}

\begin{proof} \track{For any uniform matroid $U_{k}^m$ on ground set $E$, and for any ordering $\sigma$ of elements in $E$} we have, $\sum_{i = 1}^{m} r(E_{i,\sigma}) = {k + 1 \choose 2} + k(m - k)$, \track{for prefix sets $E_{i, \sigma}$ of the ordering}. \track{If $M$ is not isomorphic to the rank-k uniform matroid $U_{k}^m$, then it must have} some subset $S \subseteq E$ of $k$ elements with rank less than $k$. As $E$ has rank at most $k$, ordering elements in $S$ first, followed by elements in $E\setminus S$ arbitrarily \track{constructs a solution with} the matroid MLOP value less than the optimal solution for $U_k^m$. The claim follows. 
 \end{proof}

 Note if $A$ is a $k \times m$ matrix, then $M[A]$ is a matroid of size $m$ and rank at most $k$. By Lemma \ref{lem:mlop_uniform}, if we solve matroid MLOP for $M[A]$, we can determine if $M[A]$ is isomorphic to $U_m^k$. By NP-hardness of uniform matroid, we have the following theorem.

\MMLOPTHM*

For matroid MLOP, the next lemma shows that solving matroid MLOP \track{for any matroid $M = (E,r)$} is as hard as solving matroid MLOP for the \track{dual matroid $M^* = (E, r^*)$}. This will be useful to show the hardness of matroid MLOP for cographic matroids.

\begin{lemma}\label{lem:matroid_MLOP_dual}
Let $M = (E,r)$ be a matroid \track{with $|E| = m$} and \track{consider an ordering} $\sigma \in \permutations{E}$, then
\begin{align*}
    \sum_{i = 1}^m r^*(E_{i,\sigma}) = {m + 1 \choose 2} - r(M)|E(M)| + \sum_{i = 1}^m r(E_{i,\sigma^*}),
\end{align*}
\sout{for the permutation $\sigma^* = |E| + 1 - \sigma \in \permutations{E}$}
\track{where $\sigma^*$ is the reverse permutation, i.e., $\sigma^* = |E| + 1 - \sigma \in \permutations{E}$}.
\end{lemma}

\begin{proof}

One can easily verify that $\sigma^* \in \permutations{E}$. As $r^*(X) = |X| - r(M) + r(E \setminus X)$ it follows,
    \begin{align*}
        \sum_{i = 1}^m r^*(E_{i,\sigma})
        &= \sum_{i = 1}^m \big( |E_{i,\sigma}| - r(M) + r(E \setminus E_{i,\sigma}) \big)\\
        &= {m + 1 \choose 2} - r(M)|E(M)| + \sum_{i = 1}^m r(E \setminus E_{i,\sigma})\\
        &= {m + 1 \choose 2} - r(M)|E(M)| + \sum_{i = 1}^m r(E_{i,\sigma^*}). 
    \end{align*}
 \end{proof}

\track{Therefore, any optimal ordering of $E$ for matroid MLOP for a given matroid $M =(E,r)$ also gives an optimal ordering for matroid MLOP on the dual matroid $M^* = (E, r^*)$.}

\begin{corollary}\label{cor:matroid_MLOP_dual}
    Matroid MLOP is NP-hard on a family of matroids ${\cal X}$ if and only if matroid MLOP is NP-hard on \track{the dual family} ${\cal X}^* = \{X^*: X \in {\cal X}\}$.
\end{corollary}

\section{Graphic matroid MLOP is NP-hard}\label{sec:Graphic_MLOP}

We next consider the complexity of graphic matroid MLOP. This turns out to be non-trivial, involving a series of reductions from minimum sum vertex cover, to minimum latency vertex cover, to weighted graphic matroid MLOP, to matroid MLOP.

To show these reductions, we first argue that an optimal chain of matroid MLOP has a useful structure of flats of the matroid, in Lemma \ref{lem:flat}. Next, in Lemma \ref{lem:MLOP_weighted}, we reduce weighted matroid MLOP to matroid MLOP. In section \ref{subsec:graphic_matroid_reduction_MLVC}, we provide a reduction from minimum latency vertex cover (MLVC) to weighted graphic matroid MLOP. Finally in \sout{section}\track{Section} \ref{subsec:MLVC_MSVC_equivalence}, we argue that MLVC and minimum set vertex cover (MSVC) are equivalent in decision form, thus completing the proof that graphic matroid MLOP is NP-hard. In Section \ref{subsec:cactus_graphs} we give an alternative characterization to matroid MLOP. In matroid MLOP we optimize over permutations of the ground set, while in this new formulation, we optimize over bases and then permutations of those bases. Using this characterization, we argue that graphic matroid MLOP for cactus graphs has a polynomial time algorithm.

\subsection{Weighted graphic matroid MLOP}

In this section, we first argue that \track{any optimal matroid MLOP solution on a ground set $E$ of size $m$ has a nice ``flat-like" structure, i.e., for any} optimal permutation $\sigma \in \permutations{E}$, the set $~\bigcup \{ E_{j,\sigma} : r(E_{j,\sigma}) \le i \}$ is a flat for all $i \in [m]$. This is a useful structural result for optimal solutions and is necessary step towards showing the hardness of graphic matroid MLOP.

\begin{lemma}\label{lem:flat}
Let $M = (E,r)$ be a matroid of size \track{$m$} and rank $k$ and let $\sigma \in \permutations{E}$ be a permutation that minimizes matroid MLOP. Then, there exists a basis $B = \{b_1, \ldots, b_k\} \in {\cal B}(M)$ and a partition $\{X_0, X_1, \ldots, X_k\}$ of $E$ such that (i) $b_i \in X_i$, and (ii) $\bigcup_{i = 0}^j X_i$ is a flat for all $0 \le j \le k$, and (iii) $\sigma(e)<\sigma(e^\prime)$ for $e \in X_i, e^\prime \in X_l$, and $i<l$.
\end{lemma} 

\begin{proof}

\track{We may suppose $k \ge 1$, as otherwise the statement is trivial.} Let \track{$\sigma \in \permutations{E}$ be a permutation that minimizes matroid MLOP and $X_i := \{ e_j : r(E_{j,\sigma}) = i \}$} for $i \ge 0$ and note that $\{X_0, X_1, \ldots, X_k\}$ partitions the ground set $E$. Furthermore $X_i \neq \emptyset$ for all $1 \le i \le k$ as for all $e \in E$ and $X \subseteq E$, we have $r(X + e) \le r(X) + 1$. \track{For each $1 \le i \le k$, let $b_i \in X_i$ be the element $e$ in $X_i$ with the lowest index $\sigma(e)$.} 

\track{For each $1 \le i \le k$, we claim $\{b_1,\ldots,b_i\}$ is an independent set. For $i = 1$, this is clear. Suppose the claim holds for all positive integers less than $j$, and $r(\{b_1,\ldots, b_j\}) = j-1$. Note that $b_j \in \closure{\{b_1, \ldots, b_{j-1}\}} = \closure{\bigcup_{i = 0}^{j-1} X_i}$. As $r(\closure{\{b_1, \ldots, b_{j-1}\}}) = j-1$, this contradicts the fact that $r(\bigcup_{i = 0}^{j-1} X_i \cup \{b_j\}) = j$. In particular, this implies that $\{b_1,\ldots, b_k\}$ is a basis of $M$ and (i) holds.}

\sout{ To construct the basis $B$ we proceed as follows, we let $b_1 \in X_1$ be arbitrary. Say for some $1 \le j < k$ we have constructed an independent set $\{b_1, \ldots, b_j\}$ such that $b_i \in X_i$ for all $1 \le i \le j$.  As $r(\bigcup_{i = 1}^{j + 1} X_{i}) = j + 1$, $\bigcup_{i = 1}^{j + 1} X_{i}$ contains an independent set of size $j+1$, say $I$. As $|\{b_1,\ldots, b_j\}| \le |I|$, we have there exists $e \in I$ such that $\{b_1, \ldots, b_j,e\}$ is also independent. As $r(\bigcup_{i = 1}^j X_i) = j$ and $\{b_1, \ldots, b_j\} \subseteq \bigcup_{i = 1}^j X_i$, we have $e \not \in \bigcup_{i = 1}^j X_i$. Thus $e \in X_{j+1}$ and we let $b_{j+1} = e$. By induction, there exists a basis $\{b_1, \ldots, b_k\}$ such that $b_i \in X_i$ for all $1 \le i \le k$.}

\sout{ We now argue for all $i$, we may have \sout{chose} \track{chosen} $b_i$ such that $\sigma(b_i) \le \sigma(e)$ for all $e \in X_i$. \sout{ By assumption of the optimality of $\sigma$, we have that $\sum_{i = 1}^m r(E_{i,\sigma}) = \min_{\pi \in \permutations{E}} \sum_{i = 1}^m r(E_{i,\pi})$.} Suppose $e,e^\prime \in X_i$, and $\sigma(e) = j < \sigma(e^\prime) = l$. As $e,e^\prime \in X_i$, we have $r(E_{j,\sigma}) = r(E_{l,\sigma}) = i$. Let $\sigma^\prime \in \permutations{E}$ such that $\sigma^\prime( e^\prime) = j$, $\sigma^{\prime}(e) = l$, and $\sigma^\prime(e^{\prime\prime}) = \sigma(e^{\prime \prime})$ for all $e^{\prime\prime} \in E \setminus \{e,e^{\prime}\}$. We now argue that $\sigma$ and $\sigma'$ have the same matroid MLOP values. To show this we only need to show $r(E_{k,\sigma^{\prime}}) = i$ for all $j \le k \le l$. By monotonicity,
$$r(E_{j,\sigma^{\prime}}) \le  r(E_{k,\sigma^{\prime}})  \le r(E_{l,\sigma^{\prime}}) = r(E_{l,\sigma}) = i.$$
Furthermore as we assumed $\sigma$ is optimal, we have $r(E_{j,\sigma^{\prime}}) \ge i$ as well. Thus without loss of generality, we may assume for all $i$, $\sigma(b_i) \le \sigma(e)$ for all $e \in X_i$.}

We now show that $\bigcup_{i=0}^j X_i$ is a flat for each $j < k$. Suppose for \track{$e' \in X_{j'}$ for $j' > j$, that $r(\bigcup_{i = 0}^j X_i \cup \{e'\}) = j$.} \sout{$\{e', b_1, \ldots b_i\}$ is dependent for some $i < j$} Let $\sigma' \in \permutations{E}$ be the permutation where we place $e'$ before $b_{j+1}$ in $\sigma$. That is, 
\vspace{-0.4em}
\begin{align*}
    \sigma'(e) = 
    \begin{cases}
    \sigma(e) &\text{ if } \sigma(e) < \sigma(b_{j+1}),\\
    \sigma(b_{j+1}) &\text{ if } e = e',\\
    \sigma(e) + 1 &\text{ if } \sigma(e) \ge \sigma(b_{j+1}) \text { and } \sigma(e) < \sigma(e^{\prime}),\\
    \sigma(e) &\text{ if } \sigma(e) \ge \sigma(e^{\prime}).
    \end{cases}.
\end{align*}

 Note $\sum_{i = 1}^m r(E_{i,\sigma'}) < \sum_{i = 1}^m r(E_{i,\sigma})$. This contradicts the optimality of $\sigma$, thus such an $e'$ cannot exist. It follows each $\bigcup_{i = 0}^j X_i$ is a flat \track{for each $j$, and hence (ii) holds. As for all $e \in X_i$ and $e' \in X_l$ with $i < l$, we have $\sigma(e) < \sigma(b_{l}) \le \sigma(e')$, thus (iii) holds as well.}
 \end{proof}

\sout{ A consequence of Lemma \ref{lem:flat} is that loops must appear at the beginning of optimal matroid MLOP chain. Combining Lemma \ref{lem:flat} with Lemma \ref{lem:matroid_MLOP_dual} we also have that coloops must appear at the end of an optimal MLOP chain.} We next introduce a weighted matroid MLOP, which given positive integer costs $c:E\rightarrow \mathbb{Z}_+$ to the elements $E$ of a matroid $M = (E, r)$, checks if there exists a permutation $\sigma\in S_E$ with weighted MLOP cost at most $K$, i.e.,   
\[ \min_{\sigma \in \permutations{E}} \sum_i r(E_{i,\sigma})c(\sigma^{-1}(i)) \le K.\] 

We now argue that weighted matroid MLOP for a matroid $M = (E,r)$ reduces to matriod MLOP as long as the total integer costs are bounded by a polynomial in $|E|$. This simply follows by duplicating an element $e\in E$ a $c(e)$ number of times, and solving unweighted matroid MLOP on the modified instance. For each duplication of $e$, $r(E_{i,\sigma})$ is counted $c(e)$ times for each duplicate for any permutation $\sigma \in \permutations{E}$.

\begin{lemma}\label{lem:MLOP_weighted}
	\sout{Let $\alg(\cdot)$ be an algorithm to solve matroid MLOP and let $\timebound(|M|)$ be the time needed by $\alg$ based on the input size of a matroid $M$.  Let $c : E \rightarrow \mathbb{N} - \{0\}$ be a given cost function and let $c(E) := \sum_{e \in E} c(e)$. Then weighted matroid MLOP with cost function $c$ can be solved in $\timebound(c(E))$ time.}

    \track{Weighted matroid MLOP with cost function $c$ can be reduced to the matroid MLOP in time $poly(|E|,c(E))$ where $c(E) := \sum_{e \in E}c(e)$.}
\end{lemma}

\begin{proof}
	Given a matroid $M = (E,r)$ with cost function \track{$c$}, let $N = (E',r')$ be the corresponding matroid where for each $e \in E(M)$ we add $c(e) - 1$ parallel elements to get $E'$. Let $m' = |E'| = c(E)$ and let $\sigma' \in \permutations{E'}$ be an optimal ordering for the \sout{modified instance}\track{matroid MLOP on $N$.} \sout{, i.e., $ \sum_{i = 1}^{m'} r(E'_{i,\sigma'}) = \min_{\pi \in \permutations{E'}} \sum_{i = 1}^{m'} r(E'_{i,\pi})$.} \sout{ By assumption,  $\alg$ solves matroid MLOP for the modified matroid $N$ in time $\timebound(c(E))$.} Since we only add parallel elements, the rank function of $M$ induces a natural rank function on $N$.

	Let $r'(N) = k$. By Lemma \ref{lem:flat}, there exists a partition $\{X_1, \ldots, X_k\}$ of $E'$ such that $\bigcup_{i = 1}^j X_i$ is a flat for all $1 \le j \le k$ and if $e \in X_i$ and $e' \in X_{\ell}$ for $i < \ell$, then $\sigma'(e) < \sigma'(e')$. Suppose $e'$ is parallel with $e$. As $\{e',e\}$ is a dependent set, we have that $e \in X_i$ if and only if $e' \in X_i$ for all $1 \le i \le k$. \track{ We now define a new ordering $\sigma''$  by rearranging the elements in $\sigma'$ such that parallel elements are grouped together by consecutive indices. As parallel elements appear in the same $X_i$, we have $\sum_{i = 1}^{m'} r(E'_{i,\sigma'}) = \sum_{i = 1}^{m'} r(E'_{i,\sigma''})$.}

    Note that $\sigma''$ induces a permutation $\sigma \in \permutations{E}$ on the original weighted matroid $M$ such that for any distinct $e,e' \in E$ we have $\sigma(e) < \sigma(e')$ if and only if $\sigma'(e) < \sigma'(e')$. Note then as parallel elements appear in the same partition set $X_i$ we have, 
	\begin{align*}
	    \sum_{i = 1}^{m'} r'(E'_{i,\sigma'}) = \sum_{i = 1}^{m'} r'(E'_{i,\sigma''}) = \sum_{i = 1}^m r(E_{i,\sigma})c(\sigma^{-1}(i)). 
	\end{align*}

	\track{Thus, the optimal weighted matroid MLOP value for $M$ with cost function $c$ is at most the optimal matroid MLOP value on $N$.} Furthermore, one can easily verify that if $\sigma \in \permutations{E}$ obtains the optimal weighted matroid MLOP value for the matroid $M$ with cost function $c$, there is a corresponding permutation $\sigma' \in \permutations{E'}$ that obtains the same matroid MLOP value for $N$. \track{Thus, the optimal values for both problems are equal.} \track{As we only added $c(E) - |E|$ additional elements to $N$, this is a $poly(|E|,c(E))$ time reduction.} \sout{ Thus the optimal weighted matroid MLOP value for $M$ with cost function $c$ is equal to the
	optimal matroid MLOP value of $N$ which can be found in time $\timebound(c(E))$. }
 \end{proof}

\subsection{Reducing MLVC to graphic matroid MLOP}\label{subsec:graphic_matroid_reduction_MLVC}

We now show that graphic matroid MLOP is as hard as minimum latency vertex cover \track{(MLVC). In MLVC we are given a graph $G=(V(G),E(G))$, and seek to find a permutation of vertices that minimizes the total edge cost, where the cost of each edge is the maximum label of its vertices, i.e., $$\min_{\pi\in \permutations{V(G)}} \sum_{(x,y)\in E(G)}\max\{\pi(x),\pi(y)\}.$$}

\begin{theorem}\label{theo:MLVC_graphic_MLOP_reduction}
\sout{Any instance of minimum latency vertex cover (MLVC) can be reduced polynomially in polynomial time to an instance of graphic matroid MLOP.}
\track{Minimum latency vertex cover (MLVC) problem can be reduced in polynomial time to the graphic matroid MLOP.}
\end{theorem}

\begin{proof} 
We will consider an instance of MLVC for a graph $G$, and construct an \sout{auxillary}\track{auxiliary} graph $H$ from $G$. We will then show that MLVC is equivalent to solving weighted graphic matroid MLOP on $H$ with a
specific cost function $c$. By showing a bound on the cost of edges $c(E(H)) := \sum_{e \in E(H)} c(e)$ in terms of a polynomial of \track{$|E(G)|$}, by applying Lemma \ref{lem:MLOP_weighted}, \track{we }will complete the reduction.

\noindent 
{\bf (a) Construction of the graphic matroid MLOP instance:} Let $G$ be the given graph with $n$ vertices and $m$ edges. We may assume without loss of generality, $G$ has no isolated vertices, as otherwise an optimal MLVC solution assigns isolated vertices last which play no role in the MLVC cost. \track{Therefore, we have that $n\leq 2m$, by counting the endpoints of the edges which upper bounds the number of vertices.}

Let $H$ be a copy of $G$ with an additional vertex $z$ connected to each vertex of $G$, i.e., $V(H) = V(G) \cup \{z\}$ and $E(H) = E(G) \cup \{(z,v) : v \in V(G)\}$. \track{Let $T$ be the spanning tree of $H$ with} $E(T) = \{(z,v) : v \in V(G)\}$. \track{Therefore, $H$ has $n+1$ vertices, and $m+n$ edges.} \track{Let $\eta := 9m^2 + 2$ and define $c(\cdot)$ to be a cost} function defined on $E(H)$, such that $c(e) = \eta$ if $e \in E(T)$, and $c(e) = 1$ otherwise. Therefore, the total cost of edges in $H$ is polynomially bounded by size of the input graph $G$: 
\begin{align*}
	    c(E(H)) &= \sum_{e \in E(H)} c(e) = \sum_{e \in E(G)} 1 + 
	    \sum_{e \in E(T)} \eta  \\
	    &= m + (9m^2 + 2)n \le m + (9m^2 + 2)2m. 
\end{align*}


 
\track{Now, for the sake of brevity, let $E := E(H)$, and $m^\prime = |E(H)|$}. Let $\sigma \in \permutations{E}$ be an optimal ordering for weighted graphic matroid MLOP over $H$ with costs $c(\cdot)$. \sout{By Lemma \ref{lem:MLOP_weighted}, $\sigma$ can be found in time $\timebound(c(E))$, which by assumption, is a polynomial in $m$. }Note,
\begin{align*}
   \text{MLOP}(H,c,\sigma) &:= \sum_{i = 1}^{m'}r(E_{i,\sigma})c(\sigma^{-1}(i))\\
   &= \sum_{i=1}^{m'} r(E_{i,\sigma}) + \sum_{e \in E(T)} r(E_{\sigma(e),\sigma})(\sout{k} \eta-1). 
\end{align*}

\noindent 
{\bf (b) Optimal solutions of graphic matroid MLOP are ``good'':} \track{We now argue that optimal solutions to graphic matroid MLOP on $H$ have a particular structure. We will argue that if the weights for edges of $T$ are large enough, then analogous to Lemma \ref{lem:flat}, each edge of $T$ must belong to a different flat induced by $\sigma$. This will be useful for relating the solutions of of graphic matroid MLOP on $H$ to MLVC on $G$.}

Let a permutation $\pi \in \permutations{E}$ be {\it good} if its prefix sets in the ordering has no two edges of $T$ induce the same rank, i.e., $r(E_{i,\pi})\neq r(E_{j,\pi})$ \track{for all distinct edges} $\pi^{-1}(i),\pi^{-1}(j) \in E(T)$. We now argue that $\sigma$ is an optimal permutation for weighted graphic matroid MLOP over $H$ only if $\sigma$ is also a good permutation. To show this, we will argue \track{a stronger claim:} any good permutation must achieve a lower MLOP objective \track{on $H$} compared to any non-good permutation. 

Let $\sigma' \in \permutations{E}$ be an arbitrary good permutation, and for the sake of contradiction, assume that there exists an optimal permutation $\sigma$ which is not-good. \track{We will argue $\sigma'$ must have lower MLOP value than $\sigma$, giving a contradiction.}

Note $\text{MLOP}(H,c, \sigma') \ge \text{MLOP}(H,c,\sigma)$. Furthermore, we have that
\begin{align*}
 \sum_{i = 1}^{m'} r(E_{i,\sigma'}) \le \sum_{i = 1}^{m'} i \le (m')^2 = (m + n)^2 \le 9m^2, 
\end{align*}
as $m + n \le 3m$. The difference in MLOP objective values of $\sigma$ and $\sigma'$ is as follows,
	\begin{align*}
		0 &\geq \text{MLOP}(H,c, \sigma)- \text{MLOP}(H,c,\sigma^\prime) \\
		&=\sum_{i=1}^{m'} r(E_{i,\sigma}) - \sum_{i=1}^{m'} r(E_{i,\sigma'}) + (\sout{k} \track{\eta} -1)\big( \sum_{e \in E(T)} r(E_{\sigma(e),\sigma}) - r(E_{\sigma'(e),\sigma'})\big) \\
        &\ge -\sum_{i=1}^{m'} r(E_{i,\sigma'}) + (\sout{k} \track{\eta} -1)\big( \sum_{e \in E(T)} r(E_{\sigma(e),\sigma}) - r(E_{\sigma'(e),\sigma'})\big) \\
		&\ge -9m^2 + (\sout{k} \track{\eta} -1)\large(\sum_{e \in E(T)} r(E_{\sigma(e),\sigma}) - r(E_{\sigma'(e),\sigma'})\large). 
	\end{align*}

    As $\sout{k} \track{\eta} -1 > 9m^2$, we have that $\sum_{e \in E(T)} r(E_{\sigma(e),\sigma}) - r(E_{\sigma'(e),\sigma'}) \le 0$, otherwise the right hand side will be strictly positive (contradicting optimality of $\sigma$, and we are done). 
    
    Let $Y_i$ be the collection of prefix sets induced by $\sigma$ \sout{upto}\track{up to} the edges of $T$ such that the rank is exactly $i$, i.e., $Y_i := \{ E_{\sigma(e),\sigma} : e \in E(T) \text{ and } r(E_{\sigma(e),\sigma}) = i\}$. Note that more than one $E_{\sigma(e),e}$ might belong to $Y_i$. Furthermore we have $\{Y_i: 1 \le i \le n\}$ partitions $\{ E_{\sigma(e),\sigma} : e \in E(T) \}$. Similarly, let $Y_i' := \{ E_{\sigma'(e),\sigma'} : \sigma' \in E(T) \text{ and } r(E_{\sigma'(e),\sigma'}) = i\}$. As $\sigma'$ is a good permutation, we have that $|Y_i'| = 1$ for all $i$. Note, 
    \begin{align*}
    0 \ge \sum_{e \in E(T)} r(E_{\sigma(e),\sigma}) - r(E_{\sigma'(e),\sigma'}') = \sum_{i = 1}^n i(|Y_i| - |Y_i'|) = \sum_{i = 1}^n i(|Y_i| - 1).
    \end{align*}
    
    We first argue that $\sum_{j = 1}^i |Y_j| \le i$ for all $1 \le i \le n$. \sout{Note $U_i := \bigcup \{ E_{\sigma(e),\sigma} :  E_{\sigma(e),\sigma} \in Y_j \text{ for } 1 \le j \le i \} = E_{\sigma(e'),\sigma}$  for some $E_{\sigma(e'),\sigma} \in \bigcup_{j = 1}^i Y_j$. }\track{Note the elements of $\bigcup_{j = 1}^i Y_j$ form a chain of subsets. Define $U_i$ to be the maximum element of $\bigcup_{j = 1}^i Y_j$.} \track{In particular, as $U_i \in Y_j$ for some $j \le i$, we have $r(U_i) \le i$}. Furthermore, as $E(T)$ is an independent set, 
    \begin{align*}
        r(U_i) \ge | \{ E_{\sigma(e),\sigma} : e \in E(T) \text{ and } E_{\sigma(e),\sigma} \subseteq U_i\}| = \sum_{j = 1}^i |Y_i|.
    \end{align*}
    Thus we have $\sum_{j = 1}^i |Y_j| \le i$. Note as $\{Y_i: 1 \le i \le n\}$ partitions $\{ E_{\sigma(e),\sigma} : e \in E(T) \}$, we also have $\sum_{j = 1}^{n} |Y_j| = E(T) = n$. 
    
    We now argue that the sum $\sum_{i = 1}^n i(|Y_i| - 1)$ is minimized, i.e., $\sum_{i = 1}^ni(|Y_i| - 1) = 0$, if only if $|Y_i|  = 1$ for all $1 \le i \le n$. Suppose $|Y_j| \neq 1$ for some $j$. Let $k$ be the first index such that $|Y_k| \neq 1$. As $k \ge \sum_{i = 1}^k |Y_i| = |Y_k| + (k - 1)$, we have that $|Y_k| = 0$. As $\sum_{i = 1}^{n} |Y_i| = n$, for some $l > k$, we have that $|Y_l| > 1$. Moving an element from $Y_l$ to $Y_j$ would decrease the sum $\sum_{i = 1}^n i(|Y_i| - 1)$. Thus $\sum_{i = 1}^n i(|Y_i| - 1)$ is minimized, i.e., $\sum_{i = 1}^n i(|Y_i| - 1) = 0$, only if $|Y_i| = 1$ for all $i$. \track{ As we assumed $\sigma$ to be not good, we have $\sum_{i = 1}^n(|Y_i| -1) > 0$, contradicting the optimality of $\sigma$. Thus we may conclude any optimal permutation for graphic matroid MLOP on $H$ must also be a good permutation.}

    \noindent
    {\bf (c). Translation of optimal solutions for graphic matroid MLOP to MLVC:} \track{ We have now argued all optimal solutions on graphic matroid MLOP on $H$ are good, i.e., each edge $e \in E(T)$ has a unique labeling $r(E_{e,\sigma})$. As every vertex of $V(G)$ is incident with exactly one edge of $E(T)$, this ordering of $E(T)$ naturally induces a permutation on $V(G)$. We claim such an ordering will be an optimal MLVC ordering on $V(G)$.}
    
     Note as $T$ is a star, all other edges of $E(G) = E \setminus E(T) = E(H) \setminus E(T)$ each form a unique triangle with the edges of $T$. Let $\pi : V(G) \rightarrow [n]$ such that $\pi(v) = r(E_{(v,z),\sigma})$ where $(v,z) \in E(T)$. As $\sigma$ is a good permutation, we have that $\pi \in \permutations{V(G)}$. \sout{Note then that for all $e = (u,v) \in E(G)$, we have $r(E_{\sigma(e),\sigma}) = \max\{\pi(v), \pi(u)\}$.}

    \track{We claim that for all $e = (u,v) \in E(G)$, we have $r(E_{\sigma(e),\sigma}) = \max\{\pi(v), \pi(u)\}$ for any optimal permutation $\sigma$ of $V(H)$. If $r(E_{\sigma(e),\sigma}) > \max\{\pi(u), \pi(v)\}$, consider the ordering $\sigma^{\prime}$ in which $e$ appears before the edge $f \in E(T)$ where $r(E_{\sigma(f),\sigma}) = \max \{\pi(u),\pi(v)\}$. This would have strictly decreasing MLOP objective value, a contradiction. Now suppose $r(E_{\sigma(e),\sigma}) < \max\{\pi(v),\pi(u)\}$, then the set} 
    \begin{align*}
        A = \{e\} \cup \{f^{\prime} : f^{\prime} \in E(T) \text{ and }  r(E_{\sigma(f^{\prime}),\sigma}) \le r(E_{\sigma(e),\sigma})\}
    \end{align*}
    \track{ is an independent set of size $r(E_{\sigma(e),\sigma}) + 1$.  Now let $f \in A \setminus \{e\}$ such that $r(E_{\sigma(f),f}) = r(E_{\sigma(e),e})$ As either $E_{\sigma(f),f}$ or $E_{\sigma(e),e}$ contains $A$, we have a contradiction.}
    
    Thus we may conclude
	\[ \text{MLOP}(H,c,\sigma) = \sum_{(u,v) \in E(G)} \max\{\pi(u), \pi(v)\} + \sum_{i = 1}^n i \cdot (9m^2 + 2).\]
	\track{As $c(E(H))$ is bounded by a polynomial in $m$, by Lemma \ref{lem:MLOP_weighted}, the MLVC problem can be reduced in polynomial time to the graphic matroid MLOP.}  
	\end{proof}

\subsection{Equivalence of MLVC and MSVC in decision form}\label{subsec:MLVC_MSVC_equivalence}

The following theorem shows that solving an MLVC instance on a \track{simple} graph $G=(V,E)$ where $|V|=n$ is equivalent to solving an MSVC instance on its complement $\overline{G}= (V, E(K_n)\setminus E)$. Since MSVC is known to be NP-hard, \sout{the reduction of MLVC to MSVC (Theorem \ref{theo:MLVC_MSVC_equivalence}) along with reduction of MLVC to weighted graphic matroid MLOP (Theorem \ref{theo:MLVC_graphic_MLOP_reduction}) and equivalence of weighted instance with unweighted instances of graphic matroid MLOP (Lemma \ref{lem:MLOP_weighted}) will together show that graphic matroid MLOP is NP-hard. }\track{by Theorem \ref{theo:MLVC_graphic_MLOP_reduction}, this will imply graphic matroid MLOP is NP-hard.}

\begin{theorem}\label{theo:MLVC_MSVC_equivalence}
	Let $G$ be a \track{simple} graph on $n$ vertices. For any labeling $\pi \in \permutations{V(G)}$, the MLVC objective on $G$ corresponds to the MSVC objective on \track{its} complement graph $\overline{G}$ with a linear shift, i.e., \begin{align*}
		 \sum_{(x,y) \in E(G)} \max\{\pi(x),\pi(y)\} &= (n^3 - n)/3 - (n+1)|E(\overline{G})| \\
		 &+ \sum_{(x,y) \in E(\overline{G})} \min \{\pi'(x), \pi'(y)\},\vspace{-0.5cm} 
	\end{align*}
	where $\pi' := n+1 - \pi \in \permutations{V(\overline{G})}$.
\end{theorem}

\begin{proof}
	Note that any labeling of the vertices of a complete graph $K_n$ gives an optimal MLVC objective value of $\sum_{i = 1}^n (i - 1)i = (n^3 - n)/3$. It follows for all $\pi \in \permutations{V(G)}$,
	\begin{align*}
		\sum_{(x,y) \in E(G)} \max\{\pi(x),\pi(y)\} + \sum_{(x,y) \in E({\overline{G}})} \max\{\pi(x),\pi(y)\} = (n^3 - n)/3.
	\end{align*}
	
	This key observation in turn gives the equivalence between MLVC and MSVC as following: 
	
	\begin{align*}
		\sum_{(x,y)\in E(G)} \max \{ \pi(x), \pi(y)\} &=  (n^3 - n)/3 - \sum_{(x,y) \in E(\overline{G})} \max \{ \pi(x), \pi(y)\} \\
		&=  (n^3 - n)/3 + \sum_{(x,y) \in E(\overline{G})}  \min \{ -\pi(x), -\pi(y)\}\\
		&= (n^3 - n)/3 -(n+1)|E(\overline{G})|\\
		&+ \sum_{(x,y) \in E(\overline{G})}\big( (n + 1) +   \min \{ -\pi(x), -\pi(y)\}\big)\\
		&=  (n^3 - n)/3 -(n+1)|E(\overline{G})|\\
		&+ \sum_{(x,y) \in E(\overline{G})}   \min \{ n + 1 -\pi(x),n + 1 -\pi(y)\} .
	\end{align*}
	
	As $\pi' = n + 1 - \pi \in \sout{\permutations{V(G)}} \track{\permutations{V(\overline{G})}}$, this completes the proof.

\end{proof}
 As MSVC is NP-hard \cite{FLT04}, we have that
\begin{corollary}\label{cor:MLVC_np_hard}
	MLVC is NP-hard.
\end{corollary}

In Theorem \ref{theo:MLVC_graphic_MLOP_reduction}, we have reduced any instance of MLVC to graphic matroid MLOP. Combining this with Corollary \ref{cor:MLVC_np_hard}, we have the promise.

\GMMLOPTHM*


In Corollary \ref{cor:matroid_MLOP_dual}, we showed if a matroid MLOP is NP-hard for a family of matroids, we have that the corresponding dual family is NP-hard as well. It follows,

\begin{corollary} \label{cor:CGMMLOP}
    Cographic matroid MLOP is NP-hard.
\end{corollary}

\subsection{Graphic matroid MLOP for cactus graphs is in P }\label{subsec:cactus_graphs}

We are now interested in further pushing the known boundaries of NP-hardness of graphic matroid MLOP, in particular show that there is a polynomial time algorithm to solve graphic matroid MLOP for cactus graphs. To achieve this we first introduce a new formulation for matroid MLOP which we believe will be of independent interest. In matroid MLOP we optimize over permutations of the ground set. In this new formulation, we first optimize over the bases of the matroid, and then \track{over all }permutations of \track{the} selected basis. \track{To see this, given a basis $B$ of a matroid and permutation $\pi \in \permutations{B}$, we construct an ordering $\sigma$ with the following rule, for each $e \not \in B$, find the minimal prefix set $X$ of $B$ such that $X \cup e$ is dependent. Place $e$ anywhere after $X$ but before the next element of $B$ in $\sigma$. If $B$ and $\pi$ are chosen as described, this will always result in an optimal MLOP permutation. Now we present this argument in detail.}

Let  $M = (E,r)$ be a loopless matroid, let $r(M) = k$ and let $\sigma \in \permutations{E}$ have optimal matroid MLOP value. By Lemma \ref{lem:flat}, there exists a partition of $E$, say $X = \{X_i : 1 \le i \le k\}$ such that $\bigcup_{i = 1}^j X_i$ is a flat for all $1 \le j \le k$ and there exists a basis $B = \{b_1, \ldots, b_k\}$ such that $b_i \in X_i$. Furthermore, we have that if $e \in X_i$ and $e' \in X_{\ell}$ for $i < \ell$, then $\sigma(e) < \sigma(e')$.

We now observe how this partition $\{X_1, \ldots, X_k\}$ interacts with the values of $r(E_{\sigma(e),\sigma})$ for optimal $\sigma$. For all $e \in E \setminus B$, $B + e$ has a unique circuit, $C(B,e)$. As $C(B,e) - e$ is an independent set, we have $| \{r(E_{\sigma(e'),\sigma}) : e' \in C(B,e) - e\}| = |C(B,e) - e|$.
As $C(B,e)$ is a dependent set and $\bigcup_{i = 1}^j X_i$ is a flat for all $1 \le j \le k$, we have that
\begin{align*}
    r(E_{\sigma(e),\sigma}) = \max \{r(E_{\sigma(e'),\sigma}) : e' \in C(B,e) - e\}.
\end{align*}
Furthermore as $B = \{b_1, \ldots, b_k\}$ with $b_i \in X_i$ for all $1 \le i \le k$, we have that there is a one-to-one correspondence between $\{b_1, \ldots, b_k\}$ and $\{r(E_{\sigma(e'),\sigma}) : e' \in B\} = \{1, \ldots, r(M)\}$. With this in mind, we define for all $\pi \in \permutations{B}$ and fundamental circuits $C(B,e)$, the set $C(B,e)_{\pi} := \{\pi(e') : e' \in C(B,e)  - e\}$. \track{Note that $C(B,e)_{\pi}$ is the set of positions in the ordering $\pi$ of the edges present in $C(B,e) - e$.}  \sout{Thus for all $e \in E \setminus B$, we have $r(E_{\sigma(e),\sigma}) = \max C(B,e)_{\pi}$.}

We now build a permutation $\sigma$ of $E$ as follows. First select a basis $B$ of the matroid, and permutation $\pi \in \permutations{B}$. Given this ordering of basis elements, we create a linear extension $\sigma$ of this order by ensuring that:
\begin{itemize}
    \item[$\bullet$] For all distinct $b,b' \in B$, $\sigma(b) < \sigma(b')$ if and only if $\pi(b) < \pi(b')$;
    \item[$\bullet$] For all $e \in E \setminus B$, if $\max C(B,e)_{\pi} = i$, then $\sigma(\pi^{-1}(i)) < \sigma(e) < \sigma(\pi^{-1}(i+1))$.
\end{itemize}
This process always constructs a permutation $\sigma \in \permutations{E}$, and if the correct basis and $\pi$ \track{are} chosen, will find the optimal matroid MLOP permutation. In particular,

\begin{proposition}\label{prop:fixed_basis_MLOP}
    Matroid MLOP is equivalent to the following problem, 
    \begin{align*}
         \min_{B \in {\cal B}(M)} \min_{\pi \in \permutations{B}} \sum_{e \in E(M) \setminus B} \max C(B,e)_{\pi}.
    \end{align*}
\end{proposition}

The characterization Proposition \ref{prop:fixed_basis_MLOP} leads to a new class of matroids in which matroid MLOP is in P. 

\SmallBasisMLOP

\begin{proof}
    By Proposition \ref{prop:fixed_basis_MLOP}, matroid MLOP for ${\cal X}$ has the following formulation, 
    \begin{align*}
         \min_{B \in {\cal B}(M)} \min_{\pi \in \permutations{B}} \sum_{e \in E(M) \setminus B} \max C(B,e)_{\pi}.
    \end{align*}
    By \cite{khachiyan2005complexity}, iterating over every basis requires $poly(m,|{\cal B}(M)|)$ time. As $|{\cal B}(M)| \le g(m)$ and $|\permutations{B}| \le (h(m))!$, simply iterating over every basis $B$  and its corresponding permutations will solve matroid MLOP for ${\cal X}$ in time $O(g(m) \cdot poly(m,g(m)) \cdot (h(m))!)$. 
 \end{proof}

\sout{We will use this new formulation}\track{We will now use Proposition \ref{prop:fixed_basis_MLOP}} to solve graphic matroid MLOP for cactus graphs. We will first argue that the selection of spanning tree is arbitrary in finding an optimal solution for cactus graphs. Then we \sout{an}order greedily with respect to the size of the circuits of the graph to find an optimal solution. 

\CacMLOPTHM*

\begin{proof}
    \track{ Let $G$ be a cactus graph. We may assume $G$ is connected, as every graphic matroid $M$ has a connected graph $H$ such that $M = M[H]$. Note as $G$ is a cactus graph, each edge of $G$ belongs to at most one cycle. Our algorithm is as follows: 
    \begin{enumerate}
    \item Order the cycles by length in nondecreasing order, temporarily regarding a bridge as a cycle of infinite length.
    \item Output any linear extension that respects this prior ordering. That is, first output all edges in the shortest cycle (in any order), followed by all edges in the next shortest cycle (in any order), and so on. 
    \end{enumerate}
    
    We now show its correctness. As bridges are coloops, a straightforward consequence of Lemma \ref{lem:matroid_MLOP_dual} implies bridges must come last in an optimal order. Thus, without loss of generality we may assume $G$ is bridgeless as well. By Proposition \ref{prop:fixed_basis_MLOP}, graphic matroid MLOP can be formulated as follows,
    \vspace{-0.1cm}
   \[ \min_{T \in {\cal B}(M)} \min_{\pi \in \permutations{T}} \sum_{e \in E(M) \setminus T} \max C(T,e)_{\pi}.\]
    where $T$ is a spanning tree of $G$, which again for convenience, we regard as a set of edges. Note as $G$ is a cactus graph, the set of fundamental circuits corresponds to the set of cycles of $G$, i.e., does not depend on the choice of $T$. The algorithm to solve MLOP for $G$ is clear, first select an arbitrary spanning tree $T$, and then order the cycles of $G$ non-decreasing with respect to their lengths. Finally, choose a $\pi \in \permutations{T}$ that respects this ordering of the circuits. It is straight forward to verify that this ordering minimizes MLOP.}

    \sout{ Let $B_1, \ldots, B_k$ be an ordering of the blocks of $G$ such that $|B_i| \le |B_{i+1}|$ for all $1 \le i < k$. \mw{Note that each block } For convenience, we will furthermore regard each $B_i$ as a set of edges. By Proposition \ref{prop:fixed_basis_MLOP}, graphic matroid MLOP can be formulated as follows,}
    \sout{ Note that $T$ is a spanning tree of $G$, which again for convenience, we regard as a set of edges.}
    
    \sout{ We will first argue that for finding the optimal MLOP value, the choice of the spanning tree $T$ is arbitrary. Let $T$ be a spanning tree such that there exists a permutation belonging to  $\permutations{T}$ that minimizes the above sum. Let $T'$ be a spanning tree that does not have a permutation that obtains the minimal matroid MLOP value such that $|T' \setminus T|$ is minimal. Let $e \in T \setminus T'$, and note that $T' + e$ must contain a circuit. Let $e \in B_i$ for some $i$, and note that $B_i \subseteq T' + e$. As $T$ is a spanning tree, $B_i \not \subseteq T$. Let $e' \in B_i \setminus T$ and note then that $T' + e  - e'$ is a spanning tree. As $|(T' + e - e') \setminus T| < |T' \setminus T|$, $T' + e - e'$ has a permutation $\pi \in \permutations{T' + e - e'}$ that minimizes the MLOP objective. Let $\pi' \in \permutations{T'}$ be the permutation such that $\pi'(e') = \pi(e)$ and if restricted to $T' - e'$, we have $\pi' = \pi$. }
    \sout{
    we have a contradiction as we assume $T'$ does not have a permutation that minimizes the matroid MLOP sum.
    }

    \sout{ This is a contradiction as we assumed $T'$ does not have a permutation that minimizes the matroid MLOP sum.} \sout{Thus every spanning of tree $T$ has a permutation that minimizes the matroid MLOP value. }
    
    \sout{ Let $T$ be any spanning tree of $E(G)$. For each edge set $B_i \cap T$, we order the edges of \mw{$B_i$ as} $e_i^j$ with $1 \le j < |B_i|$.  Let $\pi = e_1^1 < \ldots < e_1^{|B_1| - 1} < e_2^1 < \ldots < e_k^{1} < \ldots < e_k^{|B_k| - 1}$. As $|B_i| \le |B_{i + 1}|$ for all $1 \le i < k$, it is straight forward to verify that $\pi$ minimizes the MLOP sum. }

 \end{proof}

\section{Approximations for minimum latency set cover (MLSC)}\label{sec:MLVC}

In Section \ref{sec:Graphic_MLOP}, we introduced the MLVC problem in a series of reductions to show the graphic MLOP is NP-hard. \sout{Now we will a randomized $\frac{4}{3}$-approximation algorithm for MLVC. }\track{Here we study its more general version MLSC, introduced by Hassin and Levin in 2005 \cite{HL05}. In Section \ref{subsec:MLSC_alg} we present a randomized factor $(2-\frac{2}{1+\ell})$-approximation algorithm for MLSC, based on techniques from scheduling theory, where $\ell$ is the size of largest input subset. Our result is better than previously best-known factor of 2 for generic instances \cite{AGY09}. In particular, our result implies a randomized factor $\frac{4}{3}$-approximation algorithm for MLVC, improving upon Azar et. al's result \cite{AGY09}. 

We also show that for $\ell$-uniform hypergraphs, the natural linear programming (LP) relaxation (see eq. (MLSC-LP)) has an integrality gap of at least $2-\frac{2}{1+\ell}$. As a special case, we show that the integrality gap for MLVC is $\frac{4}{3}$. This implies that any approximation algorithm for MLVC based on the rounding of the LP relaxation (without additional inequalities) cannot improve upon our result.}  

In Section \ref{subsec:MLVC_poly}, \track{we explore families of instances where MLVC admits polynomial time algorithms}. We show an equivalence between MLA and MLVC for regular graphs \track{in decision form}. As many classes of regular graphs have previously been studied, this yields \track{exact} polynomial time algorithms for MLVC \track{on these families of instances}, and by Theorem \ref{theo:MLVC_MSVC_equivalence}, for MSVC \track{problem for the graph complement of these families} as well.

\subsection{A randomized approximation algorithm for MLSC based on scheduling}\label{subsec:MLSC_alg}

Recall that minimum latency vertex cover is a special case of minimum latency set cover (MLSC). MLSC can be similarly defined, as in our notation for MLVC. Instead of a graph, we are given a hypergraph $H = (V,E)$  with the objective

\[
\min_{\pi \in \permutations{V(H)}} \sum_{e \in E(H)} \max_{v \in e} \pi(v)\,.
\]

The state-of-the-art approximation for MLSC is a factor $2$, using a reduction to a well studied problem in scheduling theory,
known as $1|\text{prec}| \sum w_j C_j$, or (single machine) minimum sum scheduling with precedence constraints; that is defined as follows.
The input includes a set of jobs $J$, with corresponding processing times and weights $\{p_j\}_{j \in J}, \{w_j\}_{j \in J}$, along with a partially ordered set (poset) $P$ over the jobs. We have a single machine that takes $p_j$ amount of time to process the job $j$.
A feasible schedule is one that processes job $j$ earlier than job $j'$ whenever $j <_P j'$ in the poset.
The objective is to minimize (weighted) sum of all completion times, $\sum_j w_j C_j$, where each $C_j$ is the completion time of job $j$, and is uniquely determined by the schedule and processing times.

MLSC \track{has been known to be reducible to single machine minimum sum scheduling with precedence constraints since 2005}\cite{HL05}, using a simple construction as follows.
For every vertex $v \in V$, consider a job $v \in J$ with processing time $p_v = 1$ and weight $w_v = 0$. For every hyperedge $e \in E$, consider a job $e \in J$ with processing time $p_e = 0$ and weight $w_e = 1$. The poset $P$ over the set of jobs $J = V \cup E$ is defined by all pairs $v <_P e$ such that $v \in V, e \in E$, and $v \in e$. \track{For convenience, we also have for any distinct hyperedges $e,e' \in E$ if $e \subsetneq e'$ then $e <_P e'$. Furthermore, for all multiples of the same edge in $E$, we order them as a chain in $P$ in some arbitrary manner.}  It is easy to verify the objective of this scheduling problem is equal to that of the original MLSC. Moreover, the reduction is approximation preserving, i.e., an $\alpha$-approximate solution to the scheduling instance gives an $\alpha$-approximate solution to MLSC \cite{HL05}. 

Note that the $2$-approximability of MLSC is immediate, using various 2-approxi-mations for scheduling 
\cite{chekuri1999precedence,hall1997scheduling,margot2003decompositions}. \track{Furthermore, by Proposition \ref{prop:MLSC_MSMLOP}, MLSC is an instance of monotone submodular MLOP. Thus the $(2 - 2/(|E| + 1))$ approximation of \cite{ITT12} is applicable in this case as well.}  A better \track{constant} than $2$-approximation \track{for all instances} seems unlikely, considering hardness results for the scheduling problem \cite{BK09}, or the vertex cover problem that it reduces to \cite{correa2005single,ambuhl2009single,khot2008vertex}. 
\track{We instead show an instance-dependent improvement parameterized by the maximum size of the subsets. We achieve this result by studying} the dimension of the poset and its fractional dimension (e.g., studied by \cite{hochbaum1983efficient}, \cite{AMMS11} in the context of scheduling). \track{In the rest of this section, we prove Theorem~\ref{thm:MLVC} using the scheduling algorithm by \cite{AMMS11}. }

\MLVCTHM*



We now define the fractional dimension of a poset, that was introduced by \cite{brightwell1992fractional}.
A poset $P' (<_{P^\prime})$ is an \textit{extension} of a poset $P (<_P)$, if $x <_P y$ then $x <_{P'} y$, and $P'$ is \textit{linear} if $x \ne y$ then we have $x <_{P'} y$ or  $y <_{P'} x$.
It is easy to see that the set of feasible solutions for the single machine scheduling problem are all linear extensions of the corresponding poset. 
Let $\mathcal{F} = \{\mathcal{L}_1, \cdots, \mathcal{L}_t\}$ be a multiset of linear extensions of $P$. $\mathcal{F}$ is a \textit{$k$-fold realizer} of $P$, if for every incomparable pair $(x,y)$ of $P$, there are at least $k$ linear extensions in $\mathcal{F}$ in which $y < x$. The \textit{fractional dimension} of $P$ is \track{defined as $\lim_{k\to\infty}\frac{t}{k}$, where $t$ is the size of a minimum $k$-fold realizer (note that the fractional dimension of a poset $\geq 2$ if it is not a linear order)}. Amb{\"u}hl et al. 
\cite{AMMS11} showed $1|\text{prec}| \sum w_j C_j$ can be $(2-\frac{2}{f})$-approximated, where $f$ upper bounds the fractional dimension of the corresponding poset. Specifically, they proved the following.

\begin{theorem}[\cite{AMMS11}]\label{thm:AMMS}
Given an efficient sampling algorithm for a $k$-fold realizer of $P$, of size $t$ (that is, to output each of the $\mathcal{L}_i$'s with probability at least $1/t$), the problem $1|\text{prec}| \sum w_j C_j$ has a randomized approximation algorithm of factor $2-\frac{2}{t/k}$.
\end{theorem}

 Given an oracle that outputs a random linear extension $P^\prime$ of $P$ such that \\
$
\prob{}{j <_{P^\prime} i} \ge b
$,
for every pair of incomparable jobs $(i,j)$ in $P$, Theorem~\ref{thm:AMMS} gives a $2-2b$ approximate solution to the corresponding $1|\text{prec}| \sum w_j C_j$. Let us call the sampling algorithm provided to the above theorem, a $\frac{k}{t}$-balanced linear ordering oracle for $P$. \track{We show that it is easy to construct an $\frac{1}{1+\ell}$-balanced linear ordering oracle for posets corresponding to the MLSC's reformulation to scheduling. This will result in a $(2-\frac{2}{1+\ell})$-approximation algorithm for MLSC, using the result of Amb{\"u}hl et. al  Theorem~\ref{thm:AMMS}.}

\begin{lemma} Consider an arbitrary MLSC problem defined over a hypergraph $H = (V, E)$. Let $P$ be the poset obtained from the reformulation of the MLSC instance as a scheduling problem. Then, $P$ admits a $\frac{1}{1+\ell}$-balanced linear ordering oracle, where $\ell$ is the maximum size of any hyperedge in MLSC.
\end{lemma}
\begin{proof} 
Consider the following linear extensions to the poset $P$ constructed randomly: pick any random ordering $\{v_{l_1}, v_{l_2}, \hdots, v_{l_n}\}$ of the vertices $V$ and let them appear in the schedule in this order. To schedule any hyperedge $e\in E$, insert $e$ in the ordering as soon as all its incident vertices have been scheduled. \track{If edges are scheduled concurrently, we break ties at random.} It is easy to see that this random scheduling order leads to a valid linear extension, satisfying all precedence constrains of $P$. Let's call this linear extension $P^\prime$.

Now, we claim that any random order obtained above satisfies that the probability of $j <_{P^\prime} i$ for two incomparable jobs $i,j$ of $P$ is at least $\frac{1}{1+\ell}$. To see this, note that for a pair of vertices, this trivially holds as $\prob{}{u <_{P^\prime} v} = 0.5 \ge \frac{1}{1+\ell}$ for all distinct vertices $u$ and $v$. Let us show the inequality holds for a pair of incomparable \track{hyperedges}. For an incomparable pair consisting of a vertex and a hyperedge, we overload the notation to treat any vertex as a hyperedge of size 1. We can now consider any two distinct incomparable hyperedges $e, e'$. 

Let $a = |e \setminus e'|$, let $b = |e' \setminus e|$, and let $c = |e \cap e'|$. Note that $a, b > 0$, otherwise one edge is a subset of another, i.e., they are not incomparable \sout{(as without loss of generality, we can add a precedence constraint from the subset to the superset)}.  We compute $\prob{}{e <_{P'} e'}$ conditioning on the last vertex of $e \cup e'$ with respect to the random permutation. Call this last vertex $v_{e,e'}$.
\begin{align*}
\prob{}{e <_{P^\prime} e'} &= \prob{}{e <_{P^\prime} e' | v_{e,e'} \in e \setminus e'} \cdot \prob{}{v_{e,e'} \in e \setminus e'} \\
& + \prob{}{e <_{P^\prime} e' | v_{e,e'} \in e \cap e'} \cdot \prob{}{v_{e,e'} \in e \cap e'} \\
& + \prob{}{e <_{P^\prime} e' | v_{e,e'} \in e' \setminus e} \cdot \prob{}{v_{e,e'} \in e' \setminus e}\\
&= 0 \cdot \frac{a}{a+b+c} + \frac{1}{2} \cdot \frac{c}{a+b+c} + 1 \cdot \frac{b}{a+b+c} \\
&= \frac{b + c/2}{a+b+c}.
\end{align*}

\track{We will now use the following well-known inequality: for positive numbers $\alpha, \beta, \gamma, \delta$ such that $\alpha/\beta < \gamma/\delta$,
we have
$\frac{\alpha}{\beta} < \frac{\alpha+\gamma}{\beta+\delta} < \frac{\gamma}{\delta}$.} If $c = 0$, we have $\prob{}{e <_{P^\prime} e'} = \frac{b}{a+b} \ge \frac{1}{1+\ell}$. Suppose $c > 0$, then we can write $\prob{}{e <_{P^\prime} e'} = \frac{b + c/2}{a+b+c} \ge \min\{ \frac{b}{a+b}, \frac{c/2}{c} \}$. Considering that $\frac{b}{a+b}$ is minimized at $\frac{1}{1 + \ell}$  subject to the constraints  $1 \leq a,b \leq \ell$, we have the desired lower bound on $\prob{}{e <_{P^\prime} e'}$ in both cases.

\sout{ If $c = 0$, we have $\prob{}{e <_{P'} e'} = \frac{b}{a+b} \ge \frac{1}{1+\ell}$. Suppose $c > 0$. We can write $\prob{}{e <_{P'} e'} = \frac{b + c/2}{a+b+c} \ge \min\{ \frac{b}{a+b}, \frac{c/2}{c} \}$, where we applied a well-known inequality formulated as Fact~\ref{fact:frac}.
\begin{fact}\label{fact:frac}
For positive numbers $\alpha, \beta, \gamma, \delta$ such that $\alpha/\beta < \gamma/\delta$,
we have
$\frac{\alpha}{\beta} < \frac{\alpha+\gamma}{\beta+\delta} < \frac{\gamma}{\delta}$.
\end{fact}}

 \end{proof}

\track{Therefore, we get a $\frac{1}{1+l}$-balanced linear ordering oracle for the MLSC's scheduling reformulation, which ultimately gives us a $(2-\frac{2}{1+l})$-approximation algorithm for MLSC.} 

\paragraph{Integrality Gap for $\ell$-uniform MLSC:} Next, we consider the relaxed linear program for \track{MLSC on $\ell$-uniform hypergraphs on $n$ vertices, i.e., where each hyperedge has size $\ell$.} Here, variables $u_{e,t}$ \track{represent} whether \track{a hyperedge} $e$ is still uncovered (from \track{MLSC} perspective) until time $t$, and $x_{v,t}$ indicates a vertex $v$ to be scheduled at time step $t$, when these are constrained to be integral. 


\begin{align}
\textsc{(MLSC-LP)~~~} \text{minimize \quad}\sum\limits_{e,t}   u_{e,t}  \nonumber \\
\text{subject to} \qquad 
    \sum\limits_{v}  x_{v,t} & \leq 1,   \qquad\forall~ t \in \{1, \hdots, n\}, \label{eq:lppack}\\
   u_{e,t} + \sum_{t'<t} x_{v,t^{\prime}} & \ge 1,   \qquad\forall~ v,e,t \text{ s.t. } v \in e, \label{eq:lpcovbasic}\\
   u_{e,t},\  x_{v,t} &\ge 0, \qquad \forall~ e,v,t. 
\end{align}

The constraints~\eqref{eq:lppack}~and~\eqref{eq:lpcovbasic}, respectively, ensure that at most one vertex is scheduled during each time step, and every \track{hyperedge} remains unscheduled until all incident vertices are scheduled, i.e., $u_{e,t}$ is $0$ only if all $v \in e$ are scheduled strictly before $t$. 

First we show a lower bound of \track{$2-\frac{2}{1+\ell}$} on the integrality gap, matching the approximation factor of Theorem~\ref{thm:MLVC}.

\begin{proposition}
The integrality gap of the LP relaxation for \track{MLSC on $\ell$-uniform hypergraphs} is at least \track{$2-\frac{2}{1+\ell}$}.
\end{proposition}
\begin{proof}
\track{Consider the complete $\ell$-uniform hypergraph on $n$ vertices. By a well-known binomial coefficient identity\footnote{The hockey-stick idenitity states for positive integers $\ell \le n$, $\sum_{k = \ell}^n {k \choose \ell} = {n + 1 \choose \ell + 1}$.}, any ordering on the vertices gives} the optimal objective to the combinatorial problem, which can be shown to be
$$
\track{\sum_{k = \ell}^{n} k{k-1\choose \ell-1} = \sum_{k = \ell}^{n} \ell{k\choose \ell} =\ell{n+1\choose \ell+1}}.
$$ 

For $\ell$-uniform instances, the MLSC-LP objective can be upper bounded with a uniform fractional solution, i.e., $x_{v,t} = \frac{1}{n}$ \track{and $u_{e,t} = 1 - \frac{t-1}{n}$} for all $v$, $e$, and $t$. It follows,
$$
\sum_{e,t} u_{e,t} = |E| \cdot \big(\sum_{t = 1}^{n} (1 - \frac{t-1}{n}) \big) \track{= {n\choose \ell} \cdot \frac{n+1}{2}=\frac{\ell+1}{2}{n+1\choose \ell+1}.}$$
Thus, this family of examples provides a lower bound of \track{$\frac{2\ell}{\ell+1}=2-\frac{2}{1+\ell}$} for the integrality gap.
 \end{proof}

The integrality gap of MLSC-LP is therefore at least $2 - \frac{2}{1+l}$, \track{but it can be more for certain families of graphs. We end this section by showing that the integrality gap of the MLSC-LP is exactly \track{$2-\frac{2}{1+\ell}$}, for \track{$\ell$-uniform hypergraphs where the degree of each vertex is exactly $d$. We call these hypergraphs $d$-regular $\ell$-uniform hypergraphs}. We do not know if the integrality gap for non-regular uniform hypergraphs is strictly larger than $2-\frac{2}{1+\ell}$.} 

\begin{proposition}\label{prop:mlvc_integral}
Let $H$ be any $d$-regular \track{$\ell$-uniform hypergraph} with $n$ vertices. Then the integrality gap for $H$ is at most \track{$2-\frac{2}{1+\ell}$}. 
\end{proposition}

\begin{proof}
We first show that the \track{MLSC-LP} has an  optimal objective value \track{$\frac{dn(n+1)}{2\ell}$} for any $d$-regular \track{$\ell$-uniform hypergraph} with $n$ vertices. \track{For all fixed $1\le t\le n$, }summing over constraints~\ref{eq:lpcovbasic} for all $e \in E$ and all $v\in e$, and we have:

\begin{align}
\track{\ell}\sum_{e}u_{e,t}&=\sum_e \sum_{v\in e} u_{e,t}\\
&\overset{\eqref{eq:lpcovbasic}}{\ge} \sum_e \sum_{v\in e} \left(1-\sum_{t'<t} x_{v,t'} \right) \\
&=dn-d\sum_{t'<t}\sum_{v}x_{v,t'}\\
&\overset{\eqref{eq:lppack}}{\ge} dn-d(t-1), \, \track{\text{ for all }1\le t\le n.} \label{last}
\end{align}

Now summing over \eqref{last} for $t$ \track{from 1 to $n$} we have: 
$$
\sum_{e,t}u_{e,t}\ge \track{\frac{1}{\ell}}\sum_{t=1}^n (dn-d(t-1))=\track{\frac{dn(n+1)}{2\ell}}.
$$

It is easy to see that this objective value is achieved by letting $x_{v,t}=\frac{1}{n}$ \track{ and $u_{e,t} = 1 - \frac{t-1}{n}$} for all $e,v,t$, as this makes all inequalities satisfied \track{with} equality. 

Now consider the \track{MLSC} problem. Using randomized rounding (e.g., \cite{lau2011iterative}), we will show there exists a permutation with objective value at most \track{$2-\frac{2}{1+\ell}$} of the LP optimal value. Let $\pi$ be a uniformly random permutation of vertices, i.e. $\pi(v)=k$ with probability $1/n$ for all $1\le k\le n$. Then, \track{for any hyperedge $e$ we have}
$$\track{\expec{}{\max\{\pi(v),v\in e\}}=\frac{1}{{n\choose \ell}} \sum_{k=\ell}^n k{k-1\choose \ell-1}=\frac{\ell{n+1\choose \ell+1}}{{n\choose \ell}}=\frac{\ell(n+1)}{\ell+1}}.$$
Thus, by linearity of expectation, the expectation of the objective value for \track{MLSC} is 
$$\track{\frac{dn}{\ell}\expec{}{\max\{\pi(v),v\in e\}}=\frac{dn(n+1)}{\ell+1}}.$$
Therefore, there exists a permutation with objective value at most $\track{\frac{dn(n+1)}{\ell+1}}$, which is $\track{2-\frac{2}{1+\ell}}$ of the LP optimal value. 
 \end{proof}

\subsection{Polynomial solvable instances for MLVC and MSVC}\label{subsec:MLVC_poly}

\track{We next discuss classes of instances of MLVC and MSVC that can be solved in polynomial time.} The following theorem relates the objective value of MLA with MLVC for the family of regular graphs.

\begin{theorem}\label{theo:MLA_MLVC_regular_equivalence}
	Let $G$ be a $d$-regular graph on $n$ vertices. For any labeling $\sigma \in \permutations{n}$, we have
	\begin{align*}
		  2\cdot \hspace{-4mm} \sum_{(x,y) \in E(G)} \hspace{-3mm} \max\{\pi(x),\pi(y)\} = d{n + 1 \choose 2} + \hspace{-3mm} \sum_{(x,y) \in E(G)} \hspace{-3mm} |\pi(x) - \pi(y)|.
	\end{align*}
\end{theorem}

\begin{proof}
	We have that,
	\begin{align*}
		\sum_{(x,y) \in E(G)} \hspace{-3mm} |\pi(x) - \pi(y)| &= \sum_{(x,y) \in E(G)} \hspace{-3mm} \left[ 2 \cdot \max \{\pi(x) , \pi(y)\} - \pi(x) - \pi(y) \right]\\
		&= -\sum_{v \in V(G)} \pi(v)d +  2 \cdot \hspace{-4mm} \sum_{(x,y) \in E(G)} \hspace{-3mm}  \max \{\pi(x), \pi(y)\}\\
		&= -d\sum_{i = 1}^n i +  2 \cdot \hspace{-4mm} \sum_{(x,y) \in E(G)} \hspace{-3mm}  \max \{\pi(x), \pi(y)\}\\
		&= -d{n + 1 \choose 2}+  2 \cdot \hspace{-4mm} \sum_{(x,y) \in E(G)} \hspace{-3mm} \max \{\pi(x), \pi(y)\}.
	\end{align*}
 \end{proof}

By Theorem \ref{theo:MLA_MLVC_regular_equivalence}, we have that MLA and MLVC for regular graphs are equivalent in decision form. As the family of regular graphs is closed under graph complements, we also have by Theorem \ref{theo:MLVC_MSVC_equivalence} that MSVC and MLVC for the family of regular graphs are equivalent in decision form as well. Thus we have the following,

\begin{corollary}
    For the family of regular graphs, MLA, MLVC, and MSVC are equivalent in \sout{descision}\track{decision} form.
\end{corollary}

As an illustration of the utility of Theorem \ref{theo:MLA_MLVC_regular_equivalence}, we introduce Hamming graphs $H(d,c)$, which are obtained from $d$ Cartesian graph products of the complete graph $K_c$. Motivated by designing error correcting codes, Harper
\cite{harper1964optimal} solved the MLA problem for hypercubes, i.e. $H(d,2)$ where $d$ is any positive integer. Later, Nakano
\cite{nakano2003linear} generalized this result to all Hamming graphs $H(d,c)$ where $d$ and $c$ are positive integers. As Hamming graphs are regular, we have the following corollary of Theorem \ref{theo:MLA_MLVC_regular_equivalence}.

\begin{corollary}
    MLVC is polynomial time solvable for Hamming graphs.
\end{corollary}

The literature for the MLA problem is vast and many other instances of regular graphs have been previously solved. Thus Theorem \ref{theo:MLA_MLVC_regular_equivalence}, while simple, provides a powerful tool for providing polynomial time algorithms for many families of regular graphs. Some \track{of these families of graphs}\sout{of which} include toroidal grids \cite{muradyan1980minimal}, complete $p$-partite graphs  \cite{muradyan1980problem}, and de Bruijn graphs of order 4 \cite{harper1970chassis}. This list is by no means exhaustive, and we refer the reader to the following surveys for further reading \cite{diaz2002survey,petit2013addenda,bezrukov1999edge,lai1999survey}. Furthermore by Theorem \ref{theo:MLVC_MSVC_equivalence}, the complements of these families also have polynomial time algorithms for the MSVC problem.  

\section{\track{Improved approximation for} monotone submodular MLOP}\label{sec:monotone}

Monotone submodular MLOP was introduced by Iwata et al. \cite{ITT12}, where the authors also provided a factor $(2-\frac{2}{1+|E|})$-approximation algorithm using the Lov\'{a}sz extension of submodular functions. \track{Fokkink et al. \cite{fokkink2019submodular} studied the submodular search problem, which generalizes monotone submodular MLOP, and gave an approximation factor based on the total curvature of the submodular function.} It was not known if a tighter approximation was possible. \track{They considered the greedy contraction of the principal partition induced by the submodular function, an idea that has been used as early as 1992 by Pisaruk \cite{pisaruk1992boundaries}. In this section, we give a different analysis to the same algorithm} and improve the approximation factor to 
\begin{align*}
    2-\frac{1+\ell_f}{1+|E|} \text{ where } \track{\ell_f=\frac{f(E)}{\max_{x\in E} f(\{x\})}}.
\end{align*}

Our result can be applied to special cases including matroid and graphic matroid MLOP.\sout{In particular, when the underlying graph is connected and has a linear number of edges (e.g. $d$-regular for fixed $d$), our \track{analysis} provides a strictly less than 2 constant factor approximation to graphic matroid MLOP.} \track{For general matroid MLOP, our approximation factor is $2-\frac{1+r(E)}{1+|E|}$, which \track{is strictly smaller than 2 when $r(E)=\Omega(|E|)$ (e.g., graphic matroid on sparse graphs)}. Note that both approximation factors given by \cite{fokkink2019submodular} based on total curvature and \cite{ITT12} based on Lov{\'a}sz extension are asymptotically 2 for all non-trivial instances of matroid MLOP. }

Throughout this section, let $E$ be a nonempty set of size $m$ and $f:2^E\to\R$ be a normalized ($f(\emptyset)=0$) monotone submodular set function. \track{Without loss of generality, we can also assume that the maximum minimizer of the submodular function is the empty set\footnote{We can simply contract the maximal minimizer $U$. The elements in $U$ must appear (in any order) before the other elements $E \setminus U$ in any optimal solution for MLOP on monotone submodular functions (see Appendix \ref{app:zeroset}).}, i.e., $f(S)>0$ for all $S\ne \emptyset$.} Recall from Section \ref{sec:prelims} that the steepness of a set function $f$ is defined as $\kappa_{f}=\max_{x\in E}f(\{x\})$, and linearity of $f$ is $\ell_{f}=\frac{f(E)}{\kappa_{f}}$. By submodularity and monotonicity of $f$, for all $S\subseteq T$ we have $f(T)\le f(S)+\kappa_{f} |T\setminus S|$. 

Note for any non-trivial (i.e., $f(E)>0$) normalized monotone submodular function $f:2^E\to\R$, we have $1\le \ell_{f}\le |E|$. Both of the bounds are tight, as the lower bound $\ell_{f}=1$ is attained when $f$ is the rank function on a graphic matroid with 2 vertices and $|E|$ parallel \track{edges} between them, while the upper bound $\ell_{f}=|E|$ is attained when $f(S)=|S|$ for all $S\subseteq E$. Thus, the linearity $\ell_{f}$ is a measure of how uniform and linear a submodular function is. The function will have high linearity if each singleton has approximately same function value, and the function is approximately linear, i.e., all submodular relations $f(S)+f(T)\ge f(S\cap T)+f(S\cup T)$ are close to being tight. In the special case where $f(S)$ is the rank function of some matroid, we have $\kappa_{f}=1$ and $\ell_{f}=f(E)$ (the rank of the matroid). 

In this section, we show a $(2-\frac{1+\ell_{f}}{1+|E|})$-approximation factor to monotone submodular MLOP using any linear extension of the principal partition with respect to the submodular function. Recall that a principal partition is a set of nested sets $\emptyset=\Pi_0\subsetneq \ldots \subsetneq\Pi_s=E$ ($s\geq 1$) and a set of critical values  $\lambda_0<\lambda_1< \ldots <\lambda_{s+1}$, such that for all $0\le i\le s$, $\Pi_i$ is the unique maximal optimal solution to $\min_{X\subseteq E} f(X)-\lambda |X|$, for all $\lambda\in(\lambda_i,\lambda_{i+1})$ (Section \ref{sec:prelims}).

\sout{Principal partitions reveal a flat-like structure for submodular functions, similar to the flats in an optimal matroid chain. They provide a partial ordering on elements. Our approximation relies on a linear extension of this partial order. However, the approximation bound is non-trivial due to potentially arbitrarily small fractional marginal returns on elements. See Figure \ref{fig:principal_partition} for an illustration. Here black dots represent the principal partition. Our lower and upper bounds enclose the shaded region. The red dots represents the optimal solution. Note that they cannot be below the black dots, since principal partitions are the minimizers of $f$ among sets of same size. }

\begin{theorem}\label{thm:monotone}
Let $\{\Pi_i\}_{0\le i\le s}$ be the principal partition of a non-trivial monotone submodular function $f:2^E \to\R$ satisfying $f(\emptyset)=0$. Let $\kappa_{f}=\max_{x\in E}f(\{x\})$ and $\ell_{f}=\frac{f(E)}{\kappa_{f}}$. Let $\sigma \in \permutations{E}$ be any linear extension of the principal partition, i.e., $E_{|\Pi_i|,\sigma} = \Pi_i$ for all $1\le i\le s$. Then, the MLOP objective value of $\sigma$ is at most factor $2-\frac{1+\ell_{f}}{1+|E|}$ of the optimal solution. 
\end{theorem}

Since $1\le \ell_{f}\le |E|$, our result is a refinement on the $2-\frac{2}{1+|E|}$ factor approximation of monotone submodular MLOP in~\cite{ITT12}. \track{For the lower bound, our key lemma (Lemma \ref{lem:monotone_submodular_principal_partition_lower_bound}) is a more general version of the the well-known fact (see \cite{PrincipalPartitionF,nagano2011size,fokkink2019submodular}) that any member of the principal partition $\Pi_i$ is the ``sparsest'' subset\footnote{In our notation, ``sparsest'' subsets are the ones minimizing $\frac{f(S)}{|S|}$. Such subsets are referred to as being ``densest'' by Fokkink et al. in \cite{fokkink2019submodular}, as they maximize $\frac{|S|}{f(S)}$.} in the $\Pi_{i-1}$-contracted submodular function $f_{|\Pi_{i-1}}$, i.e., $\frac{f(S)-f(\Pi_{i-1})}{|S|-|\Pi_{i-1}|}\ge\frac{f(\Pi_i)-f(\Pi_{i-1})}{|\Pi_i|-|\Pi_{i-1}|}$ for all $S\supsetneq \Pi_{i-1}$. We show, in Lemma \ref{lem:monotone_submodular_principal_partition_lower_bound}, that this algebraic statement holds for \emph{all subsets} $S$ where $|S|\ne |\Pi_{i-1}|$, allowing us to lower bound the MLOP value of an arbitrary chain. For the upper bound, we consider any MLOP solution that is a linear extension of the principal partitions. The increase of the function value can be upper bounded using $\kappa_f$, the linearity parameter of submodular function $f$, as well as the function value at the principal partitions. }



\track{See Figure \ref{fig:principal_partition} for an illustration. The horizontal axis denotes the sizes of subsets appearing in an MLOP solution, and the vertical axis denotes the cost that these subsets incur in the MLOP objective. The coordinates of the black circles are the sizes and costs of the principal partitions. Between two adjacent black circles in the figure, the lower bound is the linear segment joining them, and the upper bound is formed using two linear segments, the first with positive slope $\kappa_f$ and the second with slope 0. The red points represent subsets in an optimal MLOP solution, and we show that they always lie inside the triangular shaded regions formed by the lower and upper bounds. In particular, the principal partitions must appear in any optimal MLOP solution, which is also a consequence of Theorem 1 in \cite{fokkink2019submodular}. }

\track{The proofs for the lower and upper bounds are highly algebraic, and a lot of calculations are deferred to the appendix. One of the challenges is that for the upper bound, the difference between function values of two adjacent subsets in the principal partition may not be an integer multiple of $\kappa_f$, thus additional steps are needed to deal with rounding as the upper bound approaches each horizontal segment. }

\begin{figure}[h]
    \centering
    \scalebox{0.8}{
    	\begin{tikzpicture}[bignode/.style={shape=circle, draw=black, ultra thick, minimum size=5mm}]
		\draw[black] (0,0) -- (10.75,0);
		\draw[black] (0,0) -- (0,8);
		\draw[dotted] (0,0) -- (4,1);
		\draw[dotted] (4,0) -- (4,1);
		\draw[dotted] (0,1) -- (4,1);

            \draw[dashed,line width=0.5mm, fill=gray!30] (4,1) -- (7,3) -- (5,3) -- (4,1);
		\draw[dotted] (0,3) -- (5,3);
		\draw[dotted] (7,0) -- (7,3);
		
		\draw[dashed,line width=0.5mm, fill=gray!30] (7,3) -- (8.5,4.5) --  (7.75,4.5) --(7,3);
		\draw[dotted] (0,4.5) -- (7.75,4.5);
		\draw[dotted] (8.5,0) -- (8.5,4.5);
            \draw[dotted] (5,0) -- (5,1.9);
            \draw[dotted] (0,1.9) -- (5,1.9);
		\draw[dotted] (8.5,4.5) -- (9.5,7);
		\draw[dotted] (9.5,0) -- (9.5,7);
		\draw[dotted] (0,7) -- (9.5,7);

            \node[bignode] at (0,0){};
            \node[bignode] at (4,1){};
            \node[bignode] at (7,3){};
            \node[bignode] at (8.5,4.5){};
            \node[bignode] at (9.5,7){};
		
		\node [fill = none] (fpi1) at (-1,0) {$f(\Pi_0) = 0$};
		\node [fill = none] (fpi1) at (-1,1) {$f(\Pi_{i-1})$};
          \node [fill = none] (fpis) at (-1,1.9) {$f(S)$};
		\node [fill = none] (fpi2) at (-1,3) {$f(\Pi_{i})$};
		\node [fill = none] (fpi3) at (-1,4.5) {$f(\Pi_{i+1})$};
		\node [fill = none] (fpi4) at (-1,7) {$f(\Pi_s)=f(E)$};
		\node [fill = none] (fpi5) at (-1.5,7.75) {MLOP costs};

		\node [fill = none] (pi1) at (0,-0.4) {$|\Pi_0| = 0$};
		\node [fill = none] (pi2) at (4,-0.4) {$|\Pi_{i-1}|$};
            \node [fill = none] (pis) at (5,-0.4) {$|S|$};
		\node [fill = none] (pi3) at (7,-0.4) {$|\Pi_{i}|$};
		\node [fill = none] (pi4) at (8.5,-0.4) {$|\Pi_{i+1}|$};
		\node [fill = none] (pi4) at (10,-0.38) {$|\Pi_s| = |E|$};
		\node [fill = none] (pi6) at (5,-1) {Size of the subsets};
		
		\node [fill = none] (1) at (3.9,2.3) {slope $\kappa_f$};
		\node [fill = none] (2) at (6.7,3.9) {slope $\kappa_f$};
            \node [fill = none] (3) at (6,1.5) {lower bound};
            \node [fill = none] (4) at (5,3.3) {upper bound};
            \node [fill = none] (5) at (8.3,3.4) {lower bound};
            \node [fill = none] (6) at (7.5,4.8) {upper bound};

  \node [style=vertex,color=red] (r1) at (0, 0){};
		\node [style=vertex,color=red] (r1) at (0.5, 0.25){};
		\node [style=vertex,color=red] (r1) at (1, 0.3){};
		\node [style=vertex,color=red] (r1) at (1.5, 0.45){};
		\node [style=vertex,color=red] (r1) at (2, 0.6){};
		\node [style=vertex,color=red] (r1) at (2.5, 0.7){};
		\node [style=vertex,color=red] (r1) at (3, 0.75){};
		\node [style=vertex,color=red] (r1) at (3.5, 0.9){};
		\node [style=vertex,color=red] (r1) at (4, 1){};
		\node [style=vertex,color=red] (r2) at (4.5, 1.7){};
		\node [style=vertex,color=red] (r3) at (5, 1.9){};
		\node [style=vertex,color=red] (r4) at (5.5, 2.4){};
		\node [style=vertex,color=red] (r5) at (6, 2.6){};
		\node [style=vertex,color=red] (r6) at (6.5, 2.85){};
	    \node [style=vertex,color=red] (r7) at (7, 3){};
	    \node [style=vertex,color=red] (r8) at (7.5, 3.7){};
	    \node [style=vertex,color=red] (r9) at (8, 4.2){};
	    \node [style=vertex,color=red] (r10) at (8.5, 4.5){};
	    \node [style=vertex,color=red] (r11) at (8.8, 5.65){};
	    \node [style=vertex,color=red] (r11) at (9.1, 6.3){};
     \node [style=vertex,color=red] (r11) at (9.5, 7){};

    \matrix [draw,below right] at (0.5,6.6) {
      \node [bignode,label=right:principal partition] {}; \\
      \node [style=vertex,color=red,label=right:optimal MLOP solution] {}; \\
      };
     
	\end{tikzpicture} 
	}
    \caption{Diagram of our lower and upper bounds in grey as well as the optimal solution in red. The black circles represent the principal partitions}
    \label{fig:principal_partition}
\end{figure}

\subsection{Lower and upper bound on MLOP objective value}

Consider a monotone submodular function $f:2^E\to\R$ satisfying $f(S)=0$ if and only if $S=\emptyset$, its principal partition $\{\Pi_i\}_{0\le i\le s}$ and the corresponding critical values $\{\lambda_i\}_{1\le i\le s}$ (Section \ref{sec:prelims}). \track{The following lemma gives the relationship between the critical values and the principal partition \cite{PrincipalPartitionF}.}

\begin{lemma}\label{lem:principal_partition_critical_value}
The principal partition $\{\Pi_i\}_{0\le i\le s}$ and corresponding critical values $\{\lambda_i\}_{1\le i\le s}$ satisfy the following relation: 

$$\lambda_i=\frac{f(\Pi_i)-f(\Pi_{i-1})}{|\Pi_i|-|\Pi_{i-1}|}, \text{ for all } 1\le i\le s. $$

Furthermore, $\Pi_{i-1}$ and $\Pi_i$ are the unique minimal and maximal 
minimizers of $\min_{X\subseteq E} f(X) - \lambda_i |X|$. 
\end{lemma}

We include a proof of Lemma \ref{lem:principal_partition_critical_value} in Appendix \ref{app:crit} for completeness. It simply uses the definition of the principal partition and submodularity of the set function.


\track{The following lemma lower gives an lower bound on the function value of any subset. As mentioned before, this lemma is more general than stating that $\Pi_i\setminus\Pi_{i-1}$ is the unique maximal sparsest subset with respect to $f_{|\Pi_{i-1}}$, the $\Pi_{i-1}$-contracted submodular function. }



\begin{lemma}\label{lem:monotone_submodular_principal_partition_lower_bound}
Let $f:2^E \to \R$ be a normalized monotone submodular function with $f(S)>0$ if $S\neq \emptyset$, and principal partition $\{\Pi_i\}_{0\le i\le s}$. Let $S\subseteq E$, then
$$
\track{\frac{f(S)-f(\Pi_{i-1})}{|S|-|\Pi_{i-1}|}\ge\frac{f(\Pi_i)-f(\Pi_{i-1})}{|\Pi_i|-|\Pi_{i-1}|},}
$$

\track{for all $i$ such that $|S|\ne |\Pi_{i-1}|$. }
\end{lemma}

\begin{proof}
One can simply fix an arbitrary critical value $\lambda_i$, and use the fact that \\\track{$f(\Pi_{i-1})-\lambda_i |\Pi_{i-1}|\le f(S)-\lambda_i |S|$} for any $S \subseteq E$. \track{Rearranging terms we get $f(S)-f(\Pi_{i-1})\ge \lambda_i\big(|S|-|\Pi_{i-1}|\big)$. }Substituting the value of $\lambda_i =\frac{f(\Pi_i)-f(\Pi_{i-1})}{|\Pi_i|-|\Pi_{i-1}|}$ (Lemma \ref{lem:principal_partition_critical_value}) \sout{where appropriate }gives us the desired result.

\end{proof}

Using the above lemma, we can sum up appropriate bounds for each subset $E_{i,\sigma}$ for any ordering $\sigma$, and obtain the following lower bound for monotone submodular MLOP. The proof after summation is purely algebraic manipulation, which is deferred to appendix.


\begin{proposition}\label{prop:pp_lower}
Let $f:2^E \to \R$ be a normalized monotone submodular function with $f(S)>0$ if $S\neq \emptyset$, and principal partition $\{\Pi_i\}_{0\le i\le s}$. Let $\sigma \in \permutations{E}$, then
$$
\sum_{k = 1}^m f(E_{k,\sigma}) \ge \frac{1}{2}(|E|+1)f(E)-\frac{1}{2}\sum_{i=1}^s \big(f(\Pi_i)|\Pi_{i-1}|-f(\Pi_{i-1})|\Pi_i|\big)\track{> 0}.
$$
\end{proposition}

\begin{proof}
\track{The proof is deferred to Appendix \ref{app:lower}.}
\end{proof}

For the upper bound, we require that the chain must contain all sets in principal partition, i.e. $E_{|\Pi_i|,\sigma} = \Pi_i$ for all $i$. We use the fact that each added element into the subset can increase the function value by at most $\kappa_f$ to upper bound the function value of remaining sets. Pictorially, if we start from $\Pi_{i-1}$, the upper bound starts at $f(\Pi_{i-1})$ and has slope $\kappa_f$, until it reaches $f(\Pi_i)$ where it remains flat until $\Pi_i$ (refer to Figure \ref{fig:principal_partition}). Also note that the increase of function value is integer multiple of $\kappa_f$ without additional analysis, and the rounding as function value approaches $f(\Pi_i)$ has to be taken care of. 

\begin{proposition}\label{prop:pp_upper}
Let $f:2^E \to \R$ be a normalized monotone submodular function with $f(S)>0$ if $S\neq \emptyset$, and principal partition $\{\Pi_i\}_{0\le i\le s}$.  Let $\sigma \in \permutations{E}$ be such that $E_{|\Pi_i|,\sigma} = \Pi_i$ for all $1\le i\le s$. Then the MLOP objective value for $f$ \track{with permutation $\sigma$} is at most

\begin{align*}
    &f(E)|E|-\frac{f(E)^2}{2\kappa_{f}}+\frac{f(E)}{2}\\
    &-\sum_{i=1}^s (f(E)-f(\Pi_i))(|\Pi_i|-|\Pi_{i-1}|)+\sum_{i=1}^s\frac{f(\Pi_{i-1})(f(\Pi_i)-f(\Pi_{i-1}))}{\kappa_{f}}. 
\end{align*}
\end{proposition}

\begin{proof}
\track{The proof is deferred to Appendix \ref{app:upper}.}
\end{proof}

\subsection{Proof of Improved Approximation for MLOP in Theorem \ref{thm:monotone}}

Both Proposition \ref{prop:pp_lower} and \ref{prop:pp_upper} together allow us to prove Theorem~\ref{thm:monotone}. Our goal is to show the upper bound obtained from Proposition~\ref{prop:pp_upper} is at most $2 - \frac{1+\ell_{f}}{1+ |E|}$ the lower bound obtained from Proposition~\ref{prop:pp_lower}, thus showing for a $\sigma \in \permutations{E}$ such that $E_{|\Pi_i|,\sigma} = \Pi_i$ for all $1 \le i \le s$, is our desired approximation for monotone submodular MLOP. 
First comparing the non-summation terms in Proposition~\ref{prop:pp_lower} and~\ref{prop:pp_upper} we have

\begin{align*}
    \frac{f(E)|E|-\frac{f(E)^2}{2\kappa_{f}}+\frac{f(E)}{2}}{\frac{1}{2}(|E|+1)f(E)}=\frac{2|E|-\frac{f(E)}{\kappa_{f}} + 1}{|E| + 1}  = 2 - \frac{1+\ell_{f}}{1+ |E|}. 
\end{align*}

To deal with the remaining summation terms, it suffices to prove that
\begin{align*}
    \sum_{i=1}^s \big(f(\Pi_i)|\Pi_{i-1}|-f(\Pi_{i-1})|\Pi_i|\big)&\le \sum_{i=1}^s (f(E)-f(\Pi_i))(|\Pi_i|-|\Pi_{i-1}|)\\
    &-\sum_{i=1}^s\frac{f(\Pi_{i-1})(f(\Pi_i)-f(\Pi_{i-1}))}{\kappa_{f}},
\end{align*}

i.e., the decrease of upper bound from non-summation terms is at least twice the decrease of lower bound from non-summation terms. To make computation easier we rewrite the terms using differential notation. For all $1\le i\le s$, let $\delta_i=f(\Pi_i)-f(\Pi_{i-1})$ and $\Delta_i=|\Pi_i|-|\Pi_{i-1}|$. By definition of $\kappa_{f}$, we have $0 \le \delta_i\le \kappa_{f} \Delta_i$. Note that $f(\Pi_i)=\sum_{j=1}^i \delta_j$ and $|\Pi_i|=\sum_{j=1}^i \Delta_j $. Furthermore, we have $f(\Pi_i)|\Pi_{i-1}|-f(\Pi_{i-1})|\Pi_i|=|\Pi_i|\delta_i-f(\Pi_i)\Delta_i$. 
Thus, the statement to be proved can be rewritten as

$$
\sum_{i=1}^s \big(\delta_i\sum_{j=1}^i \Delta_j-\Delta_i\sum_{j=1}^i \delta_j\big)\le\sum_{i=1}^s \big(\Delta_i\sum_{j=i+1}^s \delta_j -\frac{\delta_i}{\kappa_{f}}\sum_{j=1}^{i-1} \delta_j \big). 
$$

Suppose $s=1$, then both sides of this inequality are equal to zero.

Thus, we may assume $s\ge 2$. Rearranging terms, \track{for the left hand side} we have
$$\sum_{i=1}^s \big(\delta_i\sum_{j=1}^i \Delta_j-\Delta_i\sum_{j=1}^i \delta_j\big) = \sum_{i=1}^{s-1} \sum_{j>i}\delta_j\Delta_i-\delta_i\Delta_j,$$
and \track{the second part of right hand side can be rewritten as
$$\sum_{i=1}^s \frac{\delta_i}{\kappa_f}\sum_{j=1}^{i-1} \delta_j = \frac{1}{\kappa_f}\sum_{j=1}^{s-1} \sum_{i>j}\delta_i\delta_j = \frac{1}{\kappa_f}\sum_{i=1}^{s-1} \sum_{j>i}\delta_i\delta_j,$$ }
after exchanging summation order and changing variable names. As $\delta_j\le \kappa_{f}\Delta_j$ and hence $\sum_{i=1}^{s-1} \sum_{j>i}\delta_j\Delta_i-\delta_i\Delta_j\le \sum_{i=1}^{s-1} \sum_{j>i}\delta_j\Delta_i-\delta_i\frac{\delta_j}{\kappa_{f}}$, we have the inequality holds and thus, the proof is finished. 

Recall from Theorem \ref{thm:pp} that principal partitions $\{\Pi_i\}_{0\le i\le s}$ can be found in polynomial time. Thus, we have the following:

\MSMLOP*

Note that our analysis works for any linear extension to the partial order on subsets induced by the principal partition. It is unclear how this analysis can be extended to more structured linear extensions. We now discuss a special case of Theorem \ref{thm:monotone}, when $f$ is the rank function of a matroid $M$. Since in this case, $\ell_{f}=f(E)$, we get: 



\begin{corollary}\label{cor:matroid_MLOP_approximation}
Let $M=(E,r)$ be a matroid on ground set $E$ with rank function $r$. There exists a factor $(2-\frac{1+r(E)}{1+|E|})$-approximation algorithm to matroid MLOP on $M$ in polynomial time. 
\end{corollary} 

For graphic matroids, this improves upon the 2-factor approximation when graph \track{is connected and} has \track{a} linear \track{number} of edges. For instance, for \track{connected} $d$-regular graphs with vertex set $V$, the approximation factor is $2-\frac{2|V|}{2+d|V|}$, which is asymptotically $2-\frac{2}{d}$. 

\subsection{Application to minimum latency set cover (MLSC)}\label{subsec:MLSC_MSMLOP}

\track{

Recall that in Section \ref{subsec:MLSC_alg} we present a randomized factor $(2-\frac{2}{1+\ell})$-approximation algorithm for MLSC, where $\ell$ is the size of largest hyperedge. For the special case of MLVC the factor is $\frac{4}{3}$. In this section we make the observation that MLSC is an instance of monotone submodular MLOP, and use Theorem~\ref{thm:monotone} to show that there exists a deterministic factor $(2-\frac{\Delta+|E|}{\Delta (1+|V|)})$-approximation algorithm for MLSC, where $\Delta$ is the maximum degree of the hypergraph $H = (V,E)$. Note that for $\ell$-uniform hypergraphs this bound is never better than the one obtained in Section \ref{subsec:MLSC_alg}.

Recall that in MLSC, we are given a hypergraph $H=(V,E)$ with the objective

\[
\min_{\pi \in \permutations{V(H)}} \sum_{e \in E} \max_{v \in e} \pi(v)\,.
\]
In other words, we minimize over all permutations of the vertices, where the cost of each hyperedge is the maximum label of all vertices in it. Throughout this section we let $n=|V|$ denote the number of vertices. 

For a fixed $\pi\in\permutations{V}$, its reverse permutation is defined as $\pi'(v)=n+1-\pi(v)$ for all $v\in V$. We now prove that the MLSC value with $\pi$ is the same as the MLOP value with $\pi'$ on a particular monotone submodular function, which shows that MLSC is an instance of monotone submodular MLOP. 

\begin{proposition}\label{prop:MLSC_MSMLOP}
    For a fixed hypergraph $H=(V,E)$, let $f$ be the set function on $V$ such that for all $S \subseteq V$, $f(S)=|\{e\in E:S\cap e\ne\emptyset\}|$. Then $f$ is a monotone submodular function satisfying $f(\emptyset)=0$. Furthermore, for all $\pi\in\permutations{V}$ we have

    $$
    \sum_{e \in E} \max_{v \in e} \pi(v)=\sum_{i=0}^n f(V_{i,\pi'}),
    $$
    where $\pi' \in\permutations{V}$ is given by $\pi'(v)=n+1-\pi(v)$.
\end{proposition}

\begin{proof}
    It is straightforward to verify that $f$ is monotone and $f(\emptyset)=0$. For submodularity, for all $S,T\subseteq V$, observe that

    \begin{align*}
        &f(S)+f(T)-f(S\cup T)-f(S\cap T)\\
        =&|\{e\in E:e\cap S\ne\emptyset,e\cap T\ne\emptyset,e\cap S\cap T=\emptyset\}|\ge 0. 
    \end{align*}

    Now for all $0\le k\le n$ let $T_k=\{e\in E: \exists v\in e,\pi'(v)\le k \}$. Then it is straightforward to verify that for all $0\le k\le n$, $f(V_{k,\sigma})=|T_k|$ and furthermore for all $e\in E$, $|\{k:e\in T_k\}|=\max_{v \in e} \pi(v)$. Therefore we have $$ \sum_{e \in E} \max_{v \in e} \pi(v)=\sum_{i=0}^n |T_i|=\sum_{i=0}^n f(V_{i,\pi'}).$$
 \end{proof}

Using Theorem \ref{thm:monotone}, we obtain the following approximation algorithm for MLSC where the factor is based on maximum degree of the hypergraph. Note in this case $\kappa_{f}=\Delta(H)=\max_{e\in E}|e|$ is the maximum degree of hypergraph $H$. 

\MLSCMSMLOP*

For comparison, in Section~\ref{sec:MLVC} we presented a randomized scheduling-based approximation algorithm for MLSC within factor $2-\frac{2}{1+\ell}$, where $\ell=\max_{e\in E}|e|$ is the size of largest hyperedge. The algorithm presented in this section is deterministic, but for uniform hypergraphs this bound is never better than the randomized algorithm based on scheduling. 
}

\section{Future directions}
\label{sec:future}



\track{We conclude this work by presenting a list of open questions that stem from this work.}



In Sections \ref{sec:MMLOP} and \ref{sec:Graphic_MLOP} we investigated the hardness of restrictions of MLOP. In particular, we showed that graphic matroid MLOP is NP-hard. \track{In Section \ref{subsec:cactus_graphs}, we saw how  matroid MLOP can be viewed as an optimization problem over the bases of the matroid. \track{ However, even when a basis is fixed, the corresponding ordering problem on the ground set of elements can be non-trivial.} In particular, in the context for graphic matroid MLOP on a connected graph $G$, consider the following optimization problem,}
    \begin{align*}
    \min_{\sigma \in S_{T}} \sum_{e \in E(G) \setminus T} \max \{\sigma(e') : e' \in C(T,e) - e\} , 
    \end{align*}
    
\noindent 
where $T$ is a \track{given (fixed)} edge set of a spanning tree of $G$, \track{and $C(T,e)$ denotes the fundamental circuit with respect to $T$ and $e$.} We saw implicitly in the reduction of MLVC to graphic matroid MLOP (Theorem~\ref{theo:MLVC_graphic_MLOP_reduction}), that MLVC reduces to this problem when $T$ is the star graph (thus, this is NP-hard). We prove in Proposition \ref{prop:fixed_basis_MLOP}, that if we allow the choice of the spanning tree $T$ to vary \track{over all spanning trees of $G$}, then the above problem is equivalent to graphic matroid MLOP. Thus,  \track{when $T$ is fixed,} this problem \track{can be viewed as} a ``fixed-basis'' restriction of graphic matroid MLOP.

\begin{restatable}{openquestion}{Ofixedbasis}
\label{open:fixed_basis}
Given a graph $G$, a spanning tree $T$ and integer $k$, consider the problem of whether there exists a permutation $\sigma$ of $E(T)$ such that $$\sum_{e \in E(G) \setminus E(T)} \max \{\sigma(e') : e' \in C(T,e) - e\} \le k.$$ For what families of trees is this problem NP-hard?
\end{restatable}
This problem \track{is known to be NP-hard only when $T$ is a star graph, and} remains open for other simple families of trees, such as in the case where $T$ is a path.

In Section~\ref{sec:MLVC}, we showed that \track{MLSC} can be \track{$(2-\frac{2}{1+\ell})$}-approximated using randomized scheduling techniques. Furthermore, we showed for \track{$\ell$-uniform regular hypergraphs, MLSC can be $(2-\frac{2}{1+\ell})$}-approximated using an LP relaxation. This question for general \track{$\ell$-uniform hypergraphs} remains open.

\begin{restatable}{openquestion}{OLP}
\label{open:LP}
 Does solving the LP relaxation provide an approximation guarantee for \track{MLSC on $\ell$-uniform hypergraphs} by a factor of \track{$2-\frac{2}{1+\ell}$}? 
\end{restatable}

In Section \ref{subsec:MLVC_poly}, we show that the MLVC, MSVC, and MLA are all equivalent problems in decision form for regular graphs. Using techniques similar to Theorem \ref{theo:MLVC_MSVC_equivalence}, we can show that the optimal value for all three problems are related by linear shifts. It is known that MSVC on regular graphs can be 4/3-approximated (see \cite{FLT04}), but we have not found a formal proof that this problem is NP-hard. Thus, the following question remains open, to the best of our knowledge.  

\begin{restatable}{openquestion}{Oregular}
\label{open:MLA_MLVC_MSVC}
Are MLA, MLVC, and MSVC NP-hard for the family of simple regular graphs?
\end{restatable}

In Section \ref{sec:monotone}, we show that monotone submodular MLOP can be approximated within factor $2-\frac{1+\ell_f}{1+|E|}$, using principal partitions. A related open question is to develop algorithms when the principal partitions are trivial, i.e., $f(S)|E|\ge f(E)|S|$ for all $S\subseteq E$. In this case, the principal partition-based algorithm studied by Fokkink et al. in \cite{fokkink2019submodular} (and by us) will simply output an arbitrary solution. 


\begin{restatable}{openquestion}{PP}
    Do there exist better polynomial time approximation algorithms for monotone submodular MLOP in the case where the function $f$ satisfies $f(\emptyset)=0$ and $f(S)|E|\ge f(E)|S|$ for all $S\subseteq E$?
\end{restatable}

In the scope of symmetric submodular MLOP, the current best known approximation factor for the special case MLA is polylogarithmic in the size of the graph, i.e., $O(\sqrt{\log n} \log \log n)$, given by Feige and Lee \cite{FL07}; see also Charikar et al. \cite{charikar2010}. For the more general problem of symmetric submodular MLOP, there is currently no known efficient approximation algorithm better than $O(|E|)$. 

\begin{restatable}{openquestion}{OSS}
\label{open:symmetric_submodular}
Can symmetric submodular MLOP over a ground set $E$ 
be approximated to a factor better than $O(|E|)$?
\end{restatable}

\section{Acknowledgements}
The authors would like to thank Nikhil Bansal for insightful discussions on
Theorem~\ref{thm:MLVC}, i.e., approximation of MLVC, \track{L{\'a}szl{\'o} V{\'e}gh for remarks on principal partitions,} and Jai Moondra for comments on a preliminary version of this paper. \track{The authors would also like to thank the anonymous referees for numerous feedback and suggestions that are of great value, in particular for pointing out that MLSC is an instance of monotone submodular MLOP.}

The second author would like to acknowledge support from NSF CRII-1850182. The last author would like to acknowledge support by an NSF Graduate Research Fellowship under Grant No. DGE-165004. The first and third author were supported by ARC-TRIAD Student Fellowships. The fourth author acknowledges support by NSF grant DMS-2151283.

\section{Competing Interests}
The authors declare that they have no competing interests. 

\bibliography{refs}

\begin{thebibliography}{10}
\providecommand{\url}[1]{#1}
\csname url@samestyle\endcsname
\providecommand{\newblock}{\relax}
\providecommand{\bibinfo}[2]{#2}
\providecommand{\BIBentrySTDinterwordspacing}{\spaceskip=0pt\relax}
\providecommand{\BIBentryALTinterwordstretchfactor}{4}
\providecommand{\BIBentryALTinterwordspacing}{\spaceskip=\fontdimen2\font plus
\BIBentryALTinterwordstretchfactor\fontdimen3\font minus
  \fontdimen4\font\relax}
\providecommand{\BIBforeignlanguage}[2]{{%
\expandafter\ifx\csname l@#1\endcsname\relax
\typeout{** WARNING: IEEEtran.bst: No hyphenation pattern has been}%
\typeout{** loaded for the language `#1'. Using the pattern for}%
\typeout{** the default language instead.}%
\else
\language=\csname l@#1\endcsname
\fi
#2}}
\providecommand{\BIBdecl}{\relax}
\BIBdecl

\bibitem{ITT12}
S.~Iwata, P.~Tetali, and P.~Tripathi, ``Approximating minimum linear ordering
  problems,'' in \emph{{APPROX}}, 2012, pp. 206--217.

\bibitem{OW02}
J.~Oxley and D.~Welsh, ``Chromatic, flow and reliability polynomials: the
  complexity of their coefficients,'' \emph{Combinatorics Probability and
  Computing}, vol.~11, no.~4, pp. 403--426, 2002.

\bibitem{tutte1959matroids}
W.~T. Tutte, ``Matroids and graphs,'' \emph{Transactions of the American
  Mathematical Society}, vol.~90, no.~3, pp. 527--552, 1959.

\bibitem{O06}
J.~G. Oxley, \emph{Matroid theory}.\hskip 1em plus 0.5em minus 0.4em\relax
  Oxford University Press, USA, 2006, vol.~3.

\bibitem{chekuri1999precedence}
C.~Chekuri and R.~Motwani, ``Precedence constrained scheduling to minimize sum
  of weighted completion times on a single machine,'' \emph{Discrete Applied
  Mathematics}, vol.~98, no. 1-2, pp. 29--38, 1999.

\bibitem{hall1997scheduling}
L.~A. Hall, A.~S. Schulz, D.~B. Shmoys, and J.~Wein, ``Scheduling to minimize
  average completion time: Off-line and on-line approximation algorithms,''
  \emph{Mathematics of Operations Research}, vol.~22, no.~3, pp. 513--544,
  1997.

\bibitem{margot2003decompositions}
F.~Margot, M.~Queyranne, and Y.~Wang, ``Decompositions, network flows, and a
  precedence constrained single-machine scheduling problem,'' \emph{Operations
  Research}, vol.~51, no.~6, pp. 981--992, 2003.

\bibitem{AMMS11}
C.~Amb{\"u}hl, M.~Mastrolilli, N.~Mutsanas, and O.~Svensson, ``On the
  approximability of single-machine scheduling with precedence constraints,''
  \emph{Mathematics of Operations Research}, vol.~36, no.~4, pp. 653--669,
  2011.

\bibitem{diaz2002survey}
J.~D{\'\i}az, J.~Petit, and M.~Serna, ``A survey of graph layout problems,''
  \emph{ACM Computing Surveys (CSUR)}, vol.~34, no.~3, pp. 313--356, 2002.

\bibitem{petit2013addenda}
J.~Petit, ``Addenda to the survey of layout problems,'' \emph{Bulletin of
  EATCS}, vol.~3, no. 105, 2013.

\bibitem{bezrukov1999edge}
S.~L. Bezrukov, ``Edge isoperimetric problems on graphs,'' \emph{Graph Theory
  and Combinatorial Biology}, vol.~7, pp. 157--197, 1999.

\bibitem{lai1999survey}
Y.-L. Lai and K.~Williams, ``A survey of solved problems and applications on
  bandwidth, edgesum, and profile of graphs,'' \emph{Journal of graph theory},
  vol.~31, no.~2, pp. 75--94, 1999.

\bibitem{PrincipalPartitionKK}
G.~Kishi and Y.~Kajitani, ``Maximally distant trees and principal partition of
  a linear graph,'' \emph{IEEE Transactions on Circuit Theory}, vol.~16, no.~3,
  pp. 323--330, 1969.

\bibitem{PrincipalPartitionF}
\BIBentryALTinterwordspacing
S.~Fujishige, \emph{Theory of Principal Partitions Revisited}.\hskip 1em plus
  0.5em minus 0.4em\relax Berlin, Heidelberg: Springer Berlin Heidelberg, 2009,
  pp. 127--162. [Online]. Available:
  \url{https://doi.org/10.1007/978-3-540-76796-1_7}
\BIBentrySTDinterwordspacing

\bibitem{pisaruk1992boundaries}
N.~Pisaruk, ``The boundaries of submodular functions,'' \emph{Computational
  mathematics and mathematical physics}, vol.~32, no.~12, pp. 1769--1783, 1992.

\bibitem{fokkink2019submodular}
R.~Fokkink, T.~Lidbetter, and L.~A. V{\'e}gh, ``On submodular search and
  machine scheduling,'' \emph{Mathematics of Operations Research}, vol.~44,
  no.~4, pp. 1431--1449, 2019.

\bibitem{harper1964optimal}
L.~H. Harper, ``Optimal assignments of numbers to vertices,'' \emph{Journal of
  the Society for Industrial and Applied Mathematics}, vol.~12, no.~1, pp.
  131--135, 1964.

\bibitem{MLATREE}
Y.~Shiloach, ``\BIBforeignlanguage{English}{A minimum linear arrangement
  algorithm for undirected trees},'' \emph{\BIBforeignlanguage{English}{SIAM
  Journal on Computing}}, vol.~8, no.~1, pp. 15--18, 02 1979.

\bibitem{chung1984optimal}
F.-R.~K. Chung, ``On optimal linear arrangements of trees,'' \emph{Computers \&
  mathematics with applications}, vol.~10, no.~1, pp. 43--60, 1984.

\bibitem{garey1974some}
M.~R. Garey, D.~S. Johnson, and L.~Stockmeyer, ``Some simplified {NP}-complete
  problems,'' in \emph{Proceedings of the sixth annual ACM symposium on Theory
  of computing}, 1974, pp. 47--63.

\bibitem{FL07}
U.~Feige and J.~R. Lee, ``An improved approximation ratio for the minimum
  linear arrangement problem,'' \emph{Information Processing Letters}, vol.
  101, no.~1, pp. 26--29, 2007.

\bibitem{charikar2010}
M.~Charikar, M.~T. Hajiaghayi, H.~Karloff, and S.~Rao, ``$\ell_2^2$ spreading
  metrics for vertex ordering problems,'' \emph{Algorithmica}, vol.~56, no.~4,
  pp. 577--604, 2010.

\bibitem{impagliazzo2001complexity}
R.~Impagliazzo and R.~Paturi, ``On the complexity of k-{SAT},'' \emph{Journal
  of Computer and System Sciences}, vol.~62, no.~2, pp. 367--375, 2001.

\bibitem{ambuhl2011inapproximability}
C.~Amb{\"u}hl, M.~Mastrolilli, and O.~Svensson, ``Inapproximability results for
  maximum edge biclique, minimum linear arrangement, and sparsest cut,''
  \emph{SIAM Journal on Computing}, vol.~40, no.~2, pp. 567--596, 2011.

\bibitem{AGY09}
Y.~Azar, I.~Gamzu, and X.~Yin, ``Multiple intents re-ranking,'' in
  \emph{Symposium on Theory of computing, {STOC}}, 2009, pp. 669--678.

\bibitem{FLT04}
U.~Feige, L.~Lov{\'a}sz, and P.~Tetali, ``Approximating min sum set cover,''
  \emph{Algorithmica}, vol.~40, no.~4, pp. 219--234, 2004.

\bibitem{HL05}
R.~Hassin and A.~Levin, ``An approximation algorithm for the minimum latency
  set cover problem,'' in \emph{European Symposium on Algorithms, ESA}, 2005,
  pp. 726--733.

\bibitem{diaz1991minsumcut}
J.~D{\'\i}az, A.~Gibbons, M.~Paterson, and J.~Toran, ``The minsumcut problem,''
  in \emph{Workshop on Algorithms and Data Structures}.\hskip 1em plus 0.5em
  minus 0.4em\relax Springer, 1991, pp. 65--79.

\bibitem{lin1994profile}
Y.~Lin and J.~Yuan, ``Profile minimization problem for matrices and graphs,''
  \emph{Acta Mathematicae Applicatae Sinica}, vol.~10, no.~1, pp. 107--112,
  1994.

\bibitem{ES75}
S.~Even and Y.~Shiloah, ``{NP}-completeness of several arrangement problems,''
  \emph{Technical Report; Israel Institute of Technology, Department of
  Computer Science}, vol.~43, 1975.

\bibitem{BBFT21}
N.~Bansal, J.~Batra, M.~Farhadi, and P.~Tetali, ``Improved approximations for
  min sum vertex cover and generalized min sum set cover,'' in
  \emph{Proceedings of the 2021 ACM-SIAM Symposium on Discrete Algorithms
  (SODA)}.\hskip 1em plus 0.5em minus 0.4em\relax SIAM, 2021, pp. 998--1005.

\bibitem{rao2005new}
S.~Rao and A.~W. Richa, ``New approximation techniques for some linear ordering
  problems,'' \emph{SIAM Journal on Computing}, vol.~34, no.~2, pp. 388--404,
  2005.

\bibitem{BM01}
S.~Burer and R.~D. Monteiro, ``A projected gradient algorithm for solving the
  maxcut {SDP} relaxation,'' \emph{Optimization methods and Software}, vol.~15,
  no. 3-4, pp. 175--200, 2001.

\bibitem{BFP06}
U.~Barenholz, U.~Feige, D.~Peleg \emph{et~al.}, ``Improved approximation for
  min-sum vertex cover,'' MCS06-07, Computer Science and Applied Mathematics,
  Tech. Rep., 2006.

\bibitem{stankovic2022some}
A.~Stankovi{\'c}, ``Some results on approximability of minimum sum vertex
  cover,'' \emph{arXiv preprint arXiv:2212.11882}, 2022.

\bibitem{BGK10}
N.~Bansal, A.~Gupta, and R.~Krishnaswamy, ``A constant factor approximation
  algorithm for generalized min-sum set cover,'' in \emph{Symposium on Discrete
  Algorithms, {SODA}}, 2010, pp. 1539--1545.

\bibitem{SW11}
M.~Skutella and D.~P. Williamson, ``A note on the generalized min-sum set cover
  problem,'' \emph{Operations Research Letters}, vol.~39, no.~6, pp. 433--436,
  2011.

\bibitem{ISV14}
S.~Im, M.~Sviridenko, and R.~van~der Zwaan, ``Preemptive and non-preemptive
  generalized min sum set cover,'' \emph{Mathematical Programming}, vol. 145,
  no. 1-2, pp. 377--401, 2014.

\bibitem{tsaparas2011selecting}
P.~Tsaparas, A.~Ntoulas, and E.~Terzi, ``Selecting a comprehensive set of
  reviews,'' in \emph{ACM SIGKDD international conference on Knowledge
  discovery and data mining}, 2011, pp. 168--176.

\bibitem{LCS16}
W.~Luo, N.~Chakraborty, and K.~Sycara, ``Distributed dynamic priority
  assignment and motion planning for multiple mobile robots with kinodynamic
  constraints,'' in \emph{2016 American Control Conference (ACC)}.\hskip 1em
  plus 0.5em minus 0.4em\relax IEEE, 2016, pp. 148--154.

\bibitem{FLV19}
R.~Fokkink, T.~Lidbetter, and L.~A. V{\'e}gh, ``On submodular search and
  machine scheduling,'' \emph{Mathematics of Operations Research}, vol.~44,
  no.~4, pp. 1431--1449, 2019.

\bibitem{happach2020min}
F.~Happach, ``Min-sum set cover, or-scheduling, and related problems,'' Ph.D.
  dissertation, Technische Universit{\"a}t M{\"u}nchen, 2020.

\bibitem{HHL20}
F.~Happach, L.~Hellerstein, and T.~Lidbetter, ``A general framework for
  approximating min sum ordering problems,'' \emph{arXiv preprint
  arXiv:2004.05954}, 2020.

\bibitem{Schrijver03}
A.~Schrijver, \emph{Combinatorial optimization}.\hskip 1em plus 0.5em minus
  0.4em\relax Springer, 2003.

\bibitem{PrincipalPartitionN}
\BIBentryALTinterwordspacing
H.~Narayanan, ``The principal lattice of partitions of a submodular function,''
  \emph{Linear Algebra and its Applications}, vol. 144, pp. 179--216, 1991.
  [Online]. Available:
  \url{https://www.sciencedirect.com/science/article/pii/002437959190070D}
\BIBentrySTDinterwordspacing

\bibitem{khachiyan2005complexity}
L.~Khachiyan, E.~Boros, K.~Elbassioni, V.~Gurvich, and K.~Makino, ``On the
  complexity of some enumeration problems for matroids,'' \emph{SIAM Journal on
  Discrete Mathematics}, vol.~19, no.~4, pp. 966--984, 2005.

\bibitem{BK09}
N.~Bansal and S.~Khot, ``Optimal long code test with one free bit,'' in
  \emph{Foundations of Computer Science, FOCS}, 2009, pp. 453--462.

\bibitem{correa2005single}
J.~R. Correa and A.~S. Schulz, ``Single-machine scheduling with precedence
  constraints,'' \emph{Mathematics of Operations Research}, vol.~30, no.~4, pp.
  1005--1021, 2005.

\bibitem{ambuhl2009single}
C.~Amb{\"u}hl and M.~Mastrolilli, ``Single machine precedence constrained
  scheduling is a vertex cover problem,'' \emph{Algorithmica}, vol.~53, no.~4,
  pp. 488--503, 2009.

\bibitem{khot2008vertex}
S.~Khot and O.~Regev, ``Vertex cover might be hard to approximate to within 2-
  $\varepsilon$,'' \emph{Journal of Computer and System Sciences}, vol.~74,
  no.~3, pp. 335--349, 2008.

\bibitem{hochbaum1983efficient}
D.~S. Hochbaum, ``Efficient bounds for the stable set, vertex cover and set
  packing problems,'' \emph{Discrete Applied Mathematics}, vol.~6, no.~3, pp.
  243--254, 1983.

\bibitem{brightwell1992fractional}
G.~R. Brightwell and E.~R. Scheinerman, ``Fractional dimension of partial
  orders,'' \emph{Order}, vol.~9, no.~2, pp. 139--158, 1992.

\bibitem{lau2011iterative}
L.~C. Lau, R.~Ravi, and M.~Singh, \emph{Iterative methods in combinatorial
  optimization}.\hskip 1em plus 0.5em minus 0.4em\relax Cambridge University
  Press, 2011, vol.~46.

\bibitem{nakano2003linear}
K.~Nakano, ``Linear layout of generalized hypercubes,'' \emph{International
  Journal of Foundations of Computer Science}, vol.~14, no.~01, pp. 137--156,
  2003.

\bibitem{muradyan1980minimal}
D.~Muradyan and T.~Piliposjan, ``Minimal numberings of vertices of a
  rectangular lattice,'' \emph{Akad. Nauk. Armjan. SRR}, vol.~1, no.~70, pp.
  21--27, 1980. In Russian.

\bibitem{muradyan1980problem}
D.~Muradyan and T.~Piliposyan, ``The problem of finding the length and width of
  the complete p-partite graph,'' \emph{Fluchen. Zapiski Erevan. Gosunivers},
  vol.~2, pp. 18--26, 1980. In Russian.

\bibitem{harper1970chassis}
L.~Harper, ``Chassis layout and isoperimetric problems,'' \emph{Jet Propulsion
  Lab. SPS}, vol.~11, pp. 37--66, 1970.

\bibitem{nagano2011size}
K.~Nagano, Y.~Kawahara, and K.~Aihara, ``Size-constrained submodular
  minimization through minimum norm base,'' in \emph{Proceedings of the 28th
  International Conference on Machine Learning (ICML-11)}, 2011, pp. 977--984.

\end{thebibliography}
\bibliographystyle{IEEEtran}

\appendix 


\section{Proofs for Section \ref{sec:monotone}}\label{app:proofs}

\subsection{Reduction to the case where $f(S)>0$ for all $S\ne\emptyset$}\label{app:zeroset}

\track{Here we formally prove the statement that in monotone submodular MLOP where $f(\emptyset)=0$, all elements with weight zero must appear (in any order) in the beginning of any optimal MLOP solution. }

\begin{lemma}\label{lemma:monotone_zeroset}
Let $f:2^E \to \R$ be a monotone submodular function with $f(\emptyset)=0$. Then there exists a unique maximal set $U$ satisfying $f(U)=0$. Furthermore, any optimal MLOP solution $\sigma \in \permutations{E}$ on $f$ must have $U$ as prefix, i.e., $E_{|U|,\sigma}=U$.  
\end{lemma}


\begin{proof}

Let $U$ be a maximal set with $f(U)=0$ and let $U'$ be any subset of $E$ such that $f(U')=0$. We claim that $U'\subseteq U$ which would show that $U$ is the unique maximal set. From submodularity we have $f(U\cup U^\prime)\le f(U)+f(U^\prime)-f(U\cap U^\prime)=-f(U\cap U^\prime)$.
From monotonicity, $f(U \cap U^\prime) \ge f(\emptyset) = 0$, therefore $f(U \cup U^\prime) = 0$.
Since $U$ is maximal, we have $U \cup U'=U$, which implies $U^\prime \subseteq U$. 

For the sake of contradiction, let $\sigma \in \permutations{E}$ be any optimal MLOP solution where $E_{|U|,\sigma} \ne U$. Let the elements in $E = \{e_1, \hdots, e_n\}$ be ordered such that $\sigma(e_j)=j$. Let $i<|U|$ be the smallest index such that $e_i \notin U$ and let $j$ be the smallest index such that $j>i$ and $e_j\in U$. Consider a new permutation $\sigma' \in \permutations{E}$ where we move $e_j$ just before $e_i$ and keep everything else unchanged. That is,
\begin{align*}
    \sigma'(e) = 
    \begin{cases}
    \sigma(e) &\text{ if } \sigma(e) < \sigma(e_i),\\
    \sigma(e_i) &\text{ if } e = e_j,\\
    \sigma(e) + 1 &\text{ if } \sigma(e) \ge \sigma(e_i) \text { and } \sigma(e) < \sigma(e_j),\\
    \sigma(e) &\text{ if } \sigma(e) \ge \sigma(e_j).
    \end{cases}.
\end{align*}
We have $E_{k,\sigma}=E_{k,\sigma'}$ for all $k<i$ and $k\ge j$.
For all $i \le k < j$, we have that $E_{k,\sigma'} = E_{k,\sigma} - e_k + e_j$. As $f(\{e_j\}) = 0$, we have by submodularity, $f(E_{k,\sigma'})\le f(E_{k,\sigma} - e_k)+f(\{e_j\})=f(E_{k,\sigma} - e_k)$. Hence by monotonicity, $f(E_{k,\sigma'})\le f(E_{k,\sigma})$ for any $i\le k<j$. Furthermore, we have $f(E_{i,\sigma})>0$ since $e_i\notin U$. As $f(E_{i,\sigma'})=0$, this shows that changing $\sigma$ to $\sigma'$ strictly decreases the MLOP value, contradicting the optimality of $\sigma$. Thus any optimal MLOP solution must contain $U$ as a prefix.  \end{proof}



\subsection{Proof for Lemma \ref{lem:principal_partition_critical_value}}\label{app:crit}

Fix $1\le i\le s$, consider $\min_{X \subseteq E} f(X) - \lambda_i |X|$. \sout{As $\{\Pi_i\}_{0 \le i \le s}$ is a principal partition, }\track{From the definition of principal partitions, $\Pi_i$ is a minimizer of $f(X)-\lambda|X|$ for all $\lambda\in(\lambda_i,\lambda_{i+1})$, and $\Pi_{i-1}$ is a minimizer of $f(X)-\lambda|X|$ for all $\lambda\in(\lambda_{i-1},\lambda_i)$.}

\track{Thus} for all $X\subseteq E$ and sufficiently small $\varepsilon > 0$, $f(\Pi_i) - (\lambda_i + \varepsilon)|\Pi_i| \le f(X) - (\lambda_i + \varepsilon) |X|$ and $f(\Pi_i) - \lambda_i |\Pi_i| \le f(X) - \lambda_i |X|+\epsilon |\Pi_i|$. Letting $\varepsilon\to 0$, we have $f(\Pi_i) - \lambda_i |\Pi_i| \le f(X) - \lambda_i |X|$ for all $X\subseteq E$.

Furthermore, for all $X\subseteq E$ and sufficiently small $\varepsilon > 0$, $f(\Pi_{i-1}) - (\lambda_i - \varepsilon)|\Pi_{i-1}| \le f(X) - (\lambda_i - \varepsilon) |X|$ and $f(\Pi_{i-1}) - \lambda_i |\Pi_{i-1}| \le f(X) - \lambda_i |X|$ + $\epsilon |\Pi_{i-1}|$. Letting $\varepsilon\to 0$, we have $f(\Pi_{i-1}) - \lambda_i |\Pi_{i-1}|
\le f(X) - \lambda_i |X|$ for all $X\subseteq E$. 

Combining these two statements we get $f(\Pi_{i-1}) - \lambda_i |\Pi_{i-1}|=f(\Pi_i) - \lambda_i |\Pi_i|$. Solving for $\lambda_i$ we obtain the desired result. \track{The proof that $\Pi_{i-1}$ and $\Pi_i$ are the unique minimal and maximal minimizers of $\min_{X\subseteq E} f(X-\lambda_i |X|$ can be found in \cite{PrincipalPartitionF}. }

\subsection{Proof for Proposition \ref{prop:pp_lower}}\label{app:lower}

\track{We first claim that for all $S\subseteq E$ and all $i$, we have

$$
f(S)\ge f(\Pi_{i-1})+\frac{f(\Pi_i)-f(\Pi_{i-1})}{|\Pi_i|-|\Pi_{i-1}|}(|S|-|\Pi_{i-1}|).
$$

If $|S|=|\Pi_{i-1}|$ then $f(S)\ge f(\Pi_{i-1})$ from the definition of principal partitions. If $|S|\ne |\Pi_{i-1}|$, this follows from rearranging terms from Lemma \ref{lem:monotone_submodular_principal_partition_lower_bound}. 
}

Note $\sum_{k = 1}^m f(E_{k,\sigma}) = \sum_{i = 1}^s \sum_{j=|\Pi_{i-1}|+1}^{|\Pi_i|} f(E_{j,\sigma})$. For each \track{$E_{j,\sigma} \subseteq E$, we apply the bound above using the unique $i$ where $|\Pi_{i-1}|<j\le |\Pi_i|$. Since $f(\Pi_i)>f(\Pi_{i-1})$, this lower bound is strictly positive, and therefore the summation in the end is also strictly positive. 

The summation of $f(E_{j,\sigma})$ for all $|\Pi_{i-1}|<j\le |\Pi_i|$ is then given by

\begin{align*}
    \sum_{j=|\Pi_{i-1}|+1}^{|\Pi_i|} f(E_{j,\sigma})\ge & f(\Pi_{i-1})\big(|\Pi_i|-|\Pi_{i-1}|\big)+\frac{f(\Pi_i)-f(\Pi_{i-1})}{|\Pi_i|-|\Pi_{i-1}|}\sum_{j=|\Pi_{i-1}|+1}^{|\Pi_i|}\big(j-|\Pi_{i-1}|\big)\\
    = & f(\Pi_{i-1})\big(|\Pi_i|-|\Pi_{i-1}|\big)+\frac{1}{2}\big(f(\Pi_i)-f(\Pi_{i-1})\big)\big(|\Pi_i|-|\Pi_{i-1}|+1\big)\\
    = & \frac{1}{2}\big(f(\Pi_i)|\Pi_i|-f(\Pi_{i-1})|\Pi_{i-1}|+f(\Pi_i)-f(\Pi_{i-1}) \big)\\
    + & \frac{1}{2}\big(f(\Pi_{i-1})|\Pi_i|-f(\Pi_i)|\Pi_{i-1}|\big)
\end{align*}

The terms are grouped in this way so that the first part telescopes and second part does not. }It follows that 


\begin{align*}
    &\sum_{k = 1}^m f(E_{k,\sigma})=\sum_{i = 1}^s \sum_{j=|\Pi_{i-1}|+1}^{|\Pi_i|} f(E_{j,\sigma})\\
                        \ge & \sum_{i=1}^s \bigg[ \frac{1}{2}\big(f(\Pi_i)|\Pi_i|-f(\Pi_{i-1})|\Pi_{i-1}|+f(\Pi_i)-f(\Pi_{i-1}) \big)\\
    +&\frac{1}{2}\big(f(\Pi_{i-1})|\Pi_i|-f(\Pi_i)|\Pi_{i-1}|\big) \bigg]\\
                        = & \frac{1}{2}(f(\Pi_s)|\Pi_s| - f(\Pi_0)|\Pi_0| + f(\Pi_s) - f(\Pi_0))\\ -&\frac{1}{2} \sum_{i=1}^s \big(f(\Pi_i)|\Pi_{i-1}|-f(\Pi_{i-1})|\Pi_i|\big)\\
                        = & \frac{1}{2}(|E|+1)f(E)-\frac{1}{2}\sum_{i=1}^s \big(f(\Pi_i)|\Pi_{i-1}|-f(\Pi_{i-1})|\Pi_i|\big).
\end{align*}

\subsection{Proof for Proposition \ref{prop:pp_upper}}\label{app:upper}

For indices $i$ and $j$  such that $1\le i\le s$ and $|\Pi_{i-1}|+1\le j\le |\Pi_i|$ we have $f(E_{j,\sigma})\le \min\{f(\Pi_{i-1})+\kappa_{f} (j-|\Pi_{i-1}|), f(\Pi_i)\}$. Let $a_i=\lfloor \frac{f(\Pi_i)-f(\Pi_{i-1})}{\kappa_{f}}\rfloor$, which is the integer multiple of $\kappa_f$ before the upper bound becomes flat $(f(\Pi_i))$. Note that we always have $0 \le a_i\le |\Pi_i|-|\Pi_{i-1}|$ by definition of $\kappa_{f}$.

Our first goal is to sum over terms between two adjacent principal partitions. We will show that

\begin{equation}\label{equa:inequality_upperbound}
\begin{aligned}
    \sum_{j=|\Pi_{i-1}|+1}^{|\Pi_i|} f(E_{j,\sigma}) \le & f(\Pi_i)(|\Pi_i|-|\Pi_{i-1}|)\\
    - & \frac{(f(\Pi_i)-f(\Pi_{i-1}))^2}{2\kappa_{f}}+\frac{f(\Pi_i)-f(\Pi_{i-1})}{2}. 
    \end{aligned}
\end{equation}

There are three cases: 

\begin{enumerate}
    \item First suppose $a_i=0$. We have $f(\Pi_i)-f(\Pi_{i-1})<\kappa_{f}$ and (\ref{equa:inequality_upperbound}) holds as the bound \\$\sum_{j=|\Pi_{i-1}|+1}^{|\Pi_i|} f(E_{j,\sigma})\le f(\Pi_i)(|\Pi_i|-|\Pi_{i-1}|)$ is true for all $a_i$.  
    \item Suppose $a_i=|\Pi_i|-|\Pi_{i-1}|$. Since $|\Pi_i|-|\Pi_{i-1}|=a_i=\lfloor \frac{f(\Pi_i)-f(\Pi_{i-1})}{\kappa_{f}}\rfloor$, we have $f(\Pi_i)-f(\Pi_{i-1})\ge \kappa_{f}(|\Pi_i|-|\Pi_{i-1}|)$. By definition the of $\kappa_f$ we have $f(\Pi_i)-f(\Pi_{i-1})\le \kappa_{f}(|\Pi_i|-|\Pi_{i-1}|)$, and therefore $f(\Pi_i)-f(\Pi_{i-1})= \kappa_{f}(|\Pi_i|-|\Pi_{i-1}|)$. Thus we have

\begin{align*}
    & \sum_{j=|\Pi_{i-1}|+1}^{|\Pi_i|} f(E_{j,\sigma})\le \sum_{j=|\Pi_{i-1}|+1}^{|\Pi_i|} \big(f(\Pi_{i-1})+\kappa_{f} (j-|\Pi_{i-1}|)\big)\\
    = & f(\Pi_{i-1})(|\Pi_i|-|\Pi_{i-1}|)\\
    + & \frac{\kappa_f(|\Pi_i|-|\Pi_{i-1}|)(|\Pi_i|-|\Pi_{i-1}|+1)}{2}\\
    = & f(\Pi_i)(|\Pi_i|-|\Pi_{i-1}|)-\kappa_f (|\Pi_i|-|\Pi_{i-1}|)^2\\
    + & \frac{\kappa_f(|\Pi_i|-|\Pi_{i-1}|)(|\Pi_i|-|\Pi_{i-1}|+1)}{2}\\
    = & f(\Pi_i)(|\Pi_i|-|\Pi_{i-1}|)- \frac{\kappa_f (|\Pi_i|-|\Pi_{i-1}|)^2}{2}+ \frac{\kappa_f(|\Pi_i|-|\Pi_{i-1}|)}{2}\\
    = & f(\Pi_i)(|\Pi_i|-|\Pi_{i-1}|)-\frac{(f(\Pi_i)-f(\Pi_{i-1}))^2}{2\kappa_{f}}+\frac{f(\Pi_i)-f(\Pi_{i-1})}{2}, 
\end{align*}

since in this case we have $f(\Pi_i)-f(\Pi_{i-1})=\kappa_{f}(|\Pi_i|-|\Pi_{i-1}|)$. Thus (\ref{equa:inequality_upperbound}) also holds when $a_i=|\Pi_i|-|\Pi_{i-1}|$.

\item Now suppose $0< a_i< |\Pi_i|-|\Pi_{i-1}|$, we have: 

\begin{align*}
&\sum_{j=|\Pi_{i-1}|+1}^{|\Pi_i|} f(E_{j,\sigma})=\sum_{j=|\Pi_{i-1}|+1}^{|\Pi_{i-1}|+a_i} f(E_{j,\sigma})+\sum_{j=|\Pi_{i-1}|+a_i+1}^{|\Pi_i|} f(E_{j,\sigma})\\
\le & \sum_{j=|\Pi_{i-1}|+1}^{|\Pi_{i-1}|+a_i} \big(f(\Pi_{i-1})+\kappa_{f} (j-|\Pi_{i-1}|)\big)+\sum_{j=|\Pi_{i-1}|+a_i+1}^{|\Pi_i|} f(\Pi_i)\\
= & a_i f(\Pi_{i-1})+\frac{\kappa_{f} a_i(a_i+1)}{2}+f(\Pi_i)(|\Pi_i|-|\Pi_{i-1}|-a_i)\\
= & f(\Pi_i)(|\Pi_i|-|\Pi_{i-1}|)+\frac{\kappa_{f} a_i(a_i+1)}{2}-a_i (f(\Pi_i)-f(\Pi_{i-1})). 
\end{align*}

Let $\eta_i := \frac{f(\Pi_i)-f(\Pi_{i-1})}{\kappa_{f}}-a_i$ be \track{the decimal part of $\frac{f(\Pi_i)-f(\Pi_{i-1})}{\kappa_{f}}$.} Substituting $a_i$ with $ \frac{f(\Pi_i)-f(\Pi_{i-1})}{\kappa_{f}} - \eta_i$ with the above inequality we have 

\begin{align*}
\sum_{j=|\Pi_{i-1}|+1}^{|\Pi_i|} f(E_{j,\sigma})&\le
f(\Pi_i)(|\Pi_i|-|\Pi_{i-1}|)\\
&+\frac{\kappa_{f}}{2}\big( \frac{f(\Pi_i)-f(\Pi_{i-1})}{\kappa_{f}} - \eta_i\big) \big(\frac{f(\Pi_i)-f(\Pi_{i-1})}{\kappa_{f}} - \eta_i+1\big)\\
&-(\frac{f(\Pi_i)-f(\Pi_{i-1})}{\kappa_{f}} - \eta_i)(f(\Pi_i)-f(\Pi_{i-1})) \\
&\le f(\Pi_i)(|\Pi_i|-|\Pi_{i-1}|) -\frac{(f(\Pi_i)-f(\Pi_{i-1}))^2}{2\kappa_{f}}\\
&+\frac{f(\Pi_i)-f(\Pi_{i-1})}{2}+\frac{\kappa_{f}}{2}(\eta_i^2-\eta_i).
\end{align*}

As $0\le \eta_i<1$ and we have $\eta_i^2-\eta_i\le 0$. It follows,

\begin{align*}
\sum_{j=|\Pi_{i-1}|+1}^{|\Pi_i|}  f(E_{j,\sigma}) \le & f(\Pi_i)(|\Pi_i|-|\Pi_{i-1}|)\\
 -&\frac{(f(\Pi_i)-f(\Pi_{i-1}))^2}{2\kappa_{f}}+\frac{f(\Pi_i)-f(\Pi_{i-1})}{2}.
\end{align*}

\end{enumerate}

Thus inequality (\ref{equa:inequality_upperbound}) always holds.

For easier manipulation and telescoping later in the computation, we rewrite the terms as $f(\Pi_i)(|\Pi_i|-|\Pi_{i-1}|)=f(E)(|\Pi_i|-|\Pi_{i-1}|)-(f(E)-f(\Pi_i))(|\Pi_i|-|\Pi_{i-1}|)$ and $(f(\Pi_i)-f(\Pi_{i-1}))^2=f(\Pi_i)^2-f(\Pi_{i-1})^2-2f(\Pi_{i-1})(f(\Pi_i)-f(\Pi_{i-1}))$. Now summing over $i$ from 1 to $s$ we have that (\ref{equa:inequality_upperbound}) implies, 

\begin{align*}
    \sum_{k = 1}^m f(E_{k,\sigma}) &= \sum_{i = 1}^s \sum_{j=|\Pi_{i-1}|+1}^{|\Pi_i|} f(E_{j,\sigma})\\
    &\le \sum_{i=1}^s f(E)(|\Pi_i|-|\Pi_{i-1}|)- \sum_{i=1}^s (f(E)-f(\Pi_i))(|\Pi_i|-|\Pi_{i-1}|)\\
    &-\sum_{i=1}^s \frac{f(\Pi_i)^2-f(\Pi_{i-1})^2}{2\kappa_{f}} +\sum_{i=1}^s \frac{f(\Pi_{i-1})(f(\Pi_i)-f(\Pi_{i-1}))}{\kappa_{f}}\\
    &+\sum_{i=1}^s \frac{f(\Pi_i)-f(\Pi_{i-1})}{2}\\
    &= f(E)|E|-\frac{f(E)^2}{2\kappa_{f}}+\frac{f(E)}{2}-\sum_{i=1}^s (f(E)-f(\Pi_i))(|\Pi_i|-|\Pi_{i-1}|)\\
    &+\sum_{i=1}^s\frac{f(\Pi_{i-1})(f(\Pi_i)-f(\Pi_{i-1}))}{\kappa_{f}}.
\end{align*}





\end{document}